\documentclass[12pt,a4paper]{article}
 
\usepackage{geometry}
\usepackage{amsmath,amssymb,amsfonts,amsthm}
\usepackage{graphicx}
\usepackage{float}
\usepackage{xcolor}
\usepackage[numbers,square,sort&compress]{natbib}
\usepackage{hyperref}
\hypersetup{colorlinks,citecolor=blue,linkcolor=blue,breaklinks=true}
\usepackage{epstopdf}
\usepackage{yhmath} 
\usepackage{ulem}
\usepackage{subfig} 
\usepackage{xspace}

\allowdisplaybreaks
\numberwithin{equation}{section} 

\newtheorem{theorem}{Theorem}[section]
\newtheorem{lemma}{Lemma}[section]

\newtheorem{proposition}[theorem]{Proposition}
\newtheorem{definition}{Definition}[section]
\newtheorem{assumption}{\textbf{Assumption}}[section]

\newtheorem{remark}{Remark}[section]

 \newcommand{\opfont}{\mathbb}
\newcommand{\be}{\ensuremath{\opfont{E}}}
 \renewcommand{\geq}{\geqslant}
\renewcommand{\leq}{\leqslant}
\newcommand{\setg}{\mathcal{G}}

\newcommand{\bp}{\mathbb{P}}
\newcommand{\R}{\mathbb{R}}
\newcommand{\ep}{\varepsilon}
\newcommand{\dd}{\operatorname{d}\! }
\newcommand{\dt}{\operatorname{d}\! t}
\newcommand{\ds}{\operatorname{d}\! s}

\newcommand{\dx}{\operatorname{d}\! x}
\newcommand{\dy}{\operatorname{d}\! y}
\newcommand{\dz}{\operatorname{d}\! z}

\newcommand{\esb}{\kappa}
\newcommand{\barg}{\overline{G}}
\newcommand{\nn}{\nonumber}
\newcommand{\tvar}{\textrm{TVaR}\xspace}
\newcommand{\ltvar}{\textrm{left TVaR}\xspace}
\newcommand{\barz}{\overline{z}}
\newcommand{\underz}{\underline{z}}
\newcommand{\hatz}{\hat{z}}

\newcommand{\lam}{\lambda}
\newcommand{\lamb}{\nu}
\newcommand{\lamup}{\hat{\lam}}

\newcommand{\underlam}{\underline{\lam}}

\begin{document}

\title{Optimal Management of DC Pension Plan with Inflation Risk and Tail VaR Constraint }
\author{
Hui Mi \thanks{School of Mathematical Sciences,
Nanjing Normal University, Nanjing, China. Email: \url{mihui@njnu.edu.cn}. }
\and
Zuo Quan Xu\thanks{Department of Applied Mathematics, The Hong Kong Polytechnic University, Kowloon, Hong Kong, China. Email: \url{maxu@polyu.edu.hk}. }
\and
Dongfang Yang\thanks{School of Mathematical Sciences,
Nanjing Normal University, Nanjing, China. Email: \url{1569643163@qq.com}. }
}

\date{}

\maketitle

\begin{abstract}
This paper investigates an optimal investment problem under the tail Value at Risk (tail VaR, also known as expected shortfall, conditional VaR, average VaR) and portfolio insurance constraints confronted by a defined-contribution pension member. The member's aim is to maximize the expected utility from the terminal wealth exceeding the minimum guarantee by investing his wealth in a cash bond, an inflation-linked bond and a stock. 
Due to the presence of the tail VaR constraint, the problem cannot be tackled by standard control tools.
We apply the Lagrange method along with quantile optimization techniques to solve the problem. Through delicate analysis, the optimal investment output in closed-form and optimal investment strategy are derived. A numerical analysis is also provided to show how the constraints impact the optimal investment output and strategy. 
\end{abstract}

\section{Introduction}

Pension funds attract extensive attentions in financial markets nowadays. There are two major kinds of pension plans: defined-contribution (DC) pension plan and defined-benefit (DB) pension plan. Their main difference lays in that the benefit of the former relies on its portfolio performance while it is fixed for the latter. Thus it is the DC plan members who bear the financial risk instead of the sponsors. Currently, the typical tendency of the retirement system is a shift from DB to DC pension plans. 

Generally, a typical pension plan may last a lengthy period of 20-40 years, thus it is crucial to consider the impact of inflation risk on its benefit. The problem of optimal management of pension fund under inflation risk has been studied intensively in the literature. For instance, Zhang et al. \cite{Z07} and Zhang and Ewald \cite{ZE10} investigate optimal investment problems faced by a DC pension fund manager under inflation risk. Han and Hung \cite{HH12} study an optimal asset allocation problem for a DC pension plan under stochastic inflation using the stochastic dynamic programming approach.

To cover the basics after retirement, it is natural to require that the pension fund must exceed a minimum guarantee at retirement. Modeled by a portfolio insurance (PI) constraint, it can provide a downside protection against poor performance in the fund. Boulier et al. \cite{B01} consider a life annuity to the retired members. Guan and Liang \cite{GL14} extend the form of their guarantee to the case with a random time of death. Chen et al. \cite{C17} define a stochastic process representing the member’s requirement of a basic standard of living which encompasses the deterministic case.

On the other hand, Value at Risk (VaR) is one of the most widely employed risk measures by financial firms and regulators. It describes the loss which probability of occurrence exceeds a given confidence level over a certain horizon. Basak and Shapiro \cite{BS01} investigate a VaR-based risk management problem in a utility maximization framework. They use martingale method to solve the problem. As they state, the agent optimally choose to invest a larger scale in risky assets to meet with the VaR constraint, which will incur even severer losses under an adverse market scenario. This is the shortcoming that VaR has long been criticized for: focusing only on the probability while ignoring the magnitude of the loss. Cuoco \cite{C08} and Chen \cite{CH18} extend the model by embedding dynamic VaR constraints and two-period VaR constraints, respectively. Chen et al. \cite{C18} find that when incorporating a VaR and a PI constraint together, the manager would take a more prudent investment behavior. Mi and Xu \cite{MX21} investigate optimal investment problems with both VaR and PI constraints under rank-dependent expected utility framework. They find that, in bad market states, the risk of the optimal investment outcome is reduced when compared to other models without or with one constraint.
Optimal investment problems using VaR as risk management tool for pension funds have also been investigated in some works; see, e.g., Guan and Liang \cite{GL16}, Dong and Zheng \cite{DZ20}. Recently, Wu et al. \cite{W22} study an optimal investment problem under both initial time and midterm VaR constraints for a DC pension plan. They show that an intermediate-time VaR constraint can effectively reduce the risk of loss in bad market states.

Other scholars and pioneers have proposed many risk measures as an alternative to the VaR. Tail VaR (\tvar, also known as expected shortfall, conditional VaR, average VaR), 
which measures the risk by averaging all VaRs above a confidence level, 
stands out due to its desirable property as a coherent risk measure and strength which can control both the size and the probability of losses. Despite its beauty, one usually cannot obtain analytical solutions for optimization problems involving \tvar by standard control tools. 

Initiated by Jin and Zhou \cite{JZ08}, quantile formulation has recently been developed mainly for analytical tackling behavioral portfolio selection problems; see He and Zhou \cite{HZ10}, Xia and Zhou \cite{XZ16}, Xu \cite{X16} and references therein. This technique is very effective in solving those sorts of portfolio optimization problems. It switches the decision variable from optimal terminal wealth, a random variable to its quantile function, which often results in a (global or piecewise) concave/convex optimization problem instead of the original non-concave/non-convex problem. Consequently, functional optimization techniques can be applied to tackle the latter quantile optimization problem. 
For instance, He et al. \cite{HE15} use this method to solve a continuous-time mean-risk portfolio selection problem, where risk is measured by weighted-VaR risk measures. Wei \cite{W18} extends it to a utility maximization framework.

In this paper, we intend to investigate the optimal allocation for DC pension plan under joint \tvar and PI constraints, which is is inspired by Chen \cite{C18} and Wu \cite{W22}. Compared to other existing portfolio selection models, our model has three key features. First, to insure the pension benefits against inflation risk, it is necessary for a pension plan member to invest in an inflation-linked bond. Second, we incorporate \ltvar constraint with PI constraint in an optimal investment problem and investigate the feasibility and well-possedness issues in detail. Third, 
we extend the application of the quantile formulation to a utility maximization problem of a DC pension plan under \ltvar constraint. 
Due to the presence of the tail VaR constraint, all three optimization problems considered in our paper are not standard concave maximization problems, so the stochastic programming method does not work here. Instead, we will apply the martingale method to derive the optimal investment strategy. 
The main distinction from \cite{C18} and \cite{W22} is that the VaR constraints they considered can write as the expectation of indicative functions while \tvar cannot. This means we cannot find the global optimizer by comparing several values as they have done. The so-called quantile formulation plays a vital role in our paper. Another useful tool to solve constrained optimization problems is the Lagrange dual approach. Our model is more complicated than theirs since there are two Lagrange multipliers to be determined. The existence of the optimal solution along with the feasibility and well-posedness issues can be technically solved by setting the initial endowment $x_0$ within reasonable ranges. The sensitivity analysis in the end shows that the increase of confidence level and reference of \ltvar constraint will lead to an enhancement of risk-seeking attitude towards the terminal wealth. It typically manifests in purchasing more risk assets. And the additional PI constraint can help ease the tail risk of the DC pension funds.

The rest of this paper is organized as follows: Section \ref{sec:problem} introduces the financial market and the DC pension plan, as well as the stochastic control problem under \tvar and PI constraints. In Section \ref{sec:mq}, 
we derive the quantile formulation of the control problem by martingale method. 
In Section \ref{sec:quantile}, we solve the quantile optimization problem and derive the optimal strategy for the original problem. 
 Section \ref{sec:numerical} presents a sensitivity analysis to study the impacts of \ltvar constraint on the optimal strategy via numerical examples. Finally, we conclude the paper in Section \ref{sec:conclude}. 

\section{Problem formulation}\label{sec:problem}
In this section, we introduce our model which is inspired by Chen et al. \cite{C17} and Wu et al. \cite{W22}. 

Consider a filtered complete probability space $(\Omega,\mathcal{F},\{\mathcal{F}_t\}_{t\in[0,\infty)},\bp)$, known as the financial market, where $\mathcal{F}_t$ is the information available up to time $t$. 
In the probability space, a standard two dimensional Brownian motion $W(t)=(W_1(t),W_2(t))^{\top}$ is defined. Here and hereafter ${}^{\top}$ denotes the transpose of a matrix or vector. 
All the processes considered in this paper are supposed to be well-defined and adapted to $\{\mathcal{F}_t\}_{t\in[0,\infty)}$.

\subsection{Financial market}

Because a pension plan generally carries out over a long period of 20 to 40 years, the participants may face a considerable uncertainty due to inflation. To reflect the inflation, we define a stochastic inflation index process $I$, which is driven by the following stochastic differential equation (SDE):
\[\frac{\dd I(t)}{I(t)}=\mu_{I}\dt+\sigma_I\dd W_1(t), \]
where $\mu_{I}$ is the instantaneous return rate and $\sigma_I>0$ is the volatility rate of the inflation index process. 

To hedge the inflation risk, we introduce an inflation-linked bond $B$ in the market. The bond offers an instantaneous return rate $r$ exceeding the inflation index return rate, so its price is driven by 
\[\frac{\dd B(t)}{B(t)}=r\dt+\frac{\dd I(t)}{I(t)}=(r+\mu_{I})\dt+\sigma_{I}\dd W_1(t).\]
We can see that the inflation index process and the inflation-linked bond are financially equivalent as they can perfectly replicate each other. 

In addition, a risk-free bond and a stock are available to invest in the market. The risk-free bond price $S_0(t)$ follows 
\[\frac{\dd S_0(t)}{S_0(t)}=r_{0}\dt,\]
where $r_{0}$ represents the short rate of the bond, and the stock price $S(t)$ follows
\[\frac{\dd S(t)}{S(t)}=\mu_{S}\dt+\sigma_{S}\Big(\rho_{IS}\dd W_1(t)+\sqrt{1-\rho_{IS}^2}\dd W_2(t)\Big),\]
where $\mu_{S}$ and $\sigma_{S}>0$ are the appreciation rate and the volatility rate of the stock, respectively; $\rho_{IS}$ is the correlation between the inflation index process and the stock. We assume that the inflation-linked bond and the stock are not perfectly correlated so that $-1<\rho_{IS}<1$. Also, it is nature to assume the instantaneous return rates of the inflation-linked bond and stock are higher than that of the risk-free bond, namely, $r+\mu_{I}>r_{0}$ and $\mu_{S}>r_{0}$.

Because the volatility matrix 
\[\sigma=\begin{pmatrix} \sigma_{I}&0\medskip\\
\sigma_{S}\rho_{IS}&\sigma_{S}\sqrt{1-\rho_{IS}^2} \end{pmatrix} \]
is nonsingular, the financial market is complete. We define the unique market price of risk $\xi$ as 
\begin{align*} 
\xi=\begin{pmatrix}\xi_1\medskip\\\xi_2\end{pmatrix}=\sigma^{-1}\begin{pmatrix}r+\mu_{I}-r_{0}\medskip\\ \mu_{S}-r_{0}\end{pmatrix}
=\begin{pmatrix}\kappa_I\medskip\\ \frac{\kappa_{S}-\rho_{IS}\kappa_{I}}{\sqrt{1-\rho_{IS}^2}}\end{pmatrix},
\end{align*}
where $\kappa_{I}=\frac{r+\mu_{I}-r_{0}}{\sigma_{I}}>0$ and $\kappa_{S}=\frac{\mu_{S}-r_{0}}{\sigma_{S}}>0$.

\subsection{DC pension plan}
We consider a representative DC pension member ``He''. His salary process $Y$ follows:
\[\frac{\dd Y(t)}{Y(t)}=\mu_{Y}\dt+\sigma_{Y}\Big(\rho_{IY}\dd W_1(t)+\sqrt{1-\rho_{IY}^2}\dd W_2(t)\Big),\ Y(0)=y_0>0,\]
where $\mu_{Y}$ and $\sigma_{Y}>0$ are the appreciation rate and the volatility rate of his salary process, and $\rho_{IY}$ is the correlation between the inflation index process and the salary process.

Assume the representative member's contribution to the pension fund is a constant percentage $c$ ($0< c \leq 1$) of his salary. The initial wealth of his pension account is $x_0>0$. 
Let $\pi_1(t)$ and $\pi_2(t)$ be the dollar amounts invested in the inflation-linked bond and the stock, respectively, at time $t$, and the rest money is invested in the risk-free bond.
We call $\pi(t)=(\pi_1(t),\pi_2(t))^{\top}$, $t\in[0,T]$, an investment strategy or portfolio. 
Then, the wealth process $X(t)$ of his pension account follows 
\begin{equation}\label{e1}
\dd X(t)=r_{0} X(t) \dt+\pi(t)^{\top}\sigma\big(\xi \dt+\dd W(t)\big)+cY(t)\dt,\ X(0)=x_0,
\end{equation}
where $W(t)=(W_1(t),W_2(t))^{\top}$.

Let $T>0$ represent the retirement date of the representative member, which is a constant. 
The PI constraint can provide a downside protection for him. It keeps the optimal terminal wealth above the minimum guarantee at the retirement time $T$. The simplest guarantee is a positive constant which represents a lump sum to the members at time $T$. In this paper we consider the minimum performance introduced by Chen et al. \cite{C17}. 
Introduce a process $a(t)$ to represent the member’s requirement of a basic standard of living during his whole life, which follows
\[\frac{\dd a(t)}{a(t)}=\mu_adt+\sigma_a\Big(\rho_{Ia}\dd W_1(t)+\sqrt{1-\rho_{Ia}^2}\dd W_2(t)\Big),\]
where $\mu_a\geq0$ and $\sigma_a\geq 0$ are the appreciation rate and the volatility rate; $\rho_{Ia}$ is the correlation between the inflation index process and the process.
Assume the death time of the member $T'$ is a known constant bigger than $T$
\footnote{One can also assume that $T'$ follows an independent distribution. This will not change the argument too much, so we assume it is a constant for the simplicity of the presentation.}. 
The minimum performance $L(T)$, 
which stands for a sum of an elementary living requirement from the retirement time $T$ to the death time $T'$, is defined as 
\[L(T)=\be_{T}\Big[\int_T^{T'}a(s)\frac{\rho(s)}{\rho(T)}ds\Big],\]
where 
\[\rho(t)=e^{-r_{0}t-\frac{1}{2}\Vert\xi\Vert^2t-\xi^{\top}W(t)}\]
is called the pricing kernel process of the market.

We assume all the market parameters are constants. Consequently, $\rho(t)$ follows a log-normal distribution.

\begin{definition}
An investment strategy $\pi(t)$, $t\in[0,T]$, is called admissible if 
\begin{enumerate}
\item It is an $\mathcal{F}$-progressively measurable process on $[0,T]$ such that $$\int_0^T||\pi(t)||^2\dt<\infty\quad\mbox{a.s.;}$$
\item The unique strong solution $X(t)$ to \eqref{e1} satisfies 
\begin{equation}\label{add1}
X(t)+c\be_t\Big[\int_t^T\frac{\rho(s)}{\rho(t)}Y(s)\ds\Big]
\geq \be_t\Big[\int_T^{T'}a(s)\frac{\rho(s)}{\rho(t)}\ds\Big],\quad t\in[0, T].
\end{equation} 
\end{enumerate}
\end{definition}
We denote the set of admissible investment strategies by $\mathcal{A}$. From now on, we only consider admissible strategies.

The first condition ensures the SDE \eqref{e1} admits a unique strong solution $X(t)$ to \eqref{e1}. The second condition says that the present value plus the future injection to the pension funding shall guarantee the elementary living requirement for the pension member after retirement. 

In our model, the pension member will consider strategies subject to the 
left tail Value at Risk constraint (\ltvar), which is defined as 
\begin{equation}\label{es1}
\tvar^{-}_{\alpha}(X)=\frac{1}{\alpha}\int_{0}^{\alpha}{\rm VaR}^{-}_{z}(X)\dz, \quad \alpha\in(0,1),
\end{equation}
where 
\begin{equation}\label{var1}
{\rm VaR}^{-}_{z}(X)=\inf\big\{x\in\R : \bp(X\leq x)\geq z\big\},\quad z\in (0,1). 
\end{equation}
It is well known that the \tvar is a coherent risk measure (Artzner et al.\cite{ADEH99}) and a convex risk measure (F$\ddot{o}$llmer and Schied \cite{FS16}). 

The pension member's objective is to find an optimal strategy to maximize the expected utility of the wealth exceeding the minimum performance at the retirement time $T$ under the joint \ltvar and PI constraints: 
\begin{equation}\label{e2}
\begin{split}
\underset{\pi\in \mathcal{A}}{\sup}\quad & \be[U(X(T)-L(T))]\\
\mathrm{s.t.}\quad &X(\cdot)\mbox{ satisfies \eqref{e1}} ,\\
&\tvar^{-}_{\alpha}(X(T)-L(T))\geq \esb,\\
&X(T)\geq L(T)+\ell,
\end{split}
\end{equation}
where $\alpha\in(0,1)$ is a constant specified exogenously, 
$U:[0,\infty)\to [0,\infty)$ is a differentiable, strictly concave function satisfying the Inada conditions
\[\lim\limits_{x\to 0^{+}}U'(x)=+\infty,\ \lim\limits_{x\to+\infty}U'(x)=0.\] 
Note $U'$ is continuous and strictly decreasing, so it has a continuous and strictly decreasing inverse function, denoted by $(U')^{-1}$.
In this model, $X(T)\geq L(T)+\ell$ is called the PI constraint and 
$\tvar^{-}_{\alpha}(X(T)-L(T))\geq \esb$ called the \ltvar constraint on the pension value at time $T$, where $\ell$ and $\esb$ are constants specified exogenously. 

\begin{remark}
If we consider trading constraints such that the investment strategy should satisfy $\pi(t)\in K$, where K is a closed convex cone, then the market becomes incomplete and there could be many pricing kernels. 
He and Zhou \cite{HZ10} put forth the concept of minimal pricing kernel. Once the minimal pricing kernel is found, one can use quantile technique in the same way as in the complete market case, and utilize the optimal decomposition theorem as a proxy of martingale representation theorem.
\end{remark}

\begin{remark}
One may be interested in the model with an upper bounded \ltvar constraint 
\[\tvar^{-}_{\alpha}(X(T)-L(T))\leq \esb,\]
instead of the lower bounded \ltvar constraint in \eqref{e2}.
Our method still works. We will point out the major difference in the subsequent argument. 
\end{remark}

 \section{Martingale method and quantile formulation}\label{sec:mq}
 In this section, we reduce the dynamic stochastic control problem \eqref{e2} to its static quantile formulation by martingale method so that we can tackle the latter in the next section. 
 
\subsection{Reduce to a static problem}
The problem \eqref{e2} is a continuous time stochastic control problem. Because it has state constraint, the standard stochastic control methods such as dynamic programming and maximum principle are hard to apply to it. To tackle it, we use another powerful method, the martingale method. This method is widely used to solve utility maximization problems in complete market. In this approach, one first turns a stochastic control problem into a static random variable optimization problem, then solves the latter by optimization techniques, and finally recovers the optimal investment strategy by the backward stochastic control theory. 

For our problem, the difficulty lies in the second step: solving the optimization problem.
This is due to the \ltvar constraint. We will overcome this difficult by the quantile optimization method in the next section.

In our model, the member contributes to the pension fund continuously, so 
the wealth process \eqref{e1} is not self-financing. To apply the martingale method, 
we first introduce an auxiliary process to obtain an equivalent problem. 
Inspired by the constraint \eqref{add1}, we introduce a process 
\[Z(t)=X(t)+D(t)-L(t),\quad t\in[0,T],\]
where $D(t)$ stands for the present value of the expected aggregated contribution from $t$ to $T$ given by 
\[D(t)=c\be_t\Big[\int_t^T\frac{\rho(s)}{\rho(t)}Y(s)\ds\Big], \]
and $L(t)$ stands for the present value of the expected future minimum performance given by
\[L(t)=\be_t\Big[\int_T^{T'}a(s)\frac{\rho(s)}{\rho(t)}\ds\Big].\]
Obviously, $Z(T)=X(T)-L(T)$ since $D(T)=0$, and the PI constraint $X(T)\geq L(T)+\ell$ is equivalent to $Z(T)\geq\ell $.

The following result characterizes the processes $Z$, $D$ and $L$.
\begin{proposition}
\begin{enumerate}
\item 
The present value of the expected future contribution is a multiple of the instantaneous contribution, that is, 
\begin{align*} 
D(t)=
\begin{cases}
 \frac{1}{\beta_{D}}(e^{\beta_{D}(T-t)}-1)cY(t), &\beta_D \neq 0,\medskip\\
 (T-t)cY(t), &\beta_D=0, 
\end{cases}
\end{align*}
where \[\beta_{D}=\mu_Y-r_{0}-\sigma_Y\rho_{IY}\xi_1-\sigma_Y\sqrt{1-\rho_{IY}^2}\xi_2.\] 
In particular, we set \[d_0:=D(0)=\frac{1}{\beta_{D}}(e^{\beta_{D} T}-1)cy_0.\]

\item The present value of the expected future minimum performance is a multiple of 
the basic standard of living process, that is, 
\begin{align*} 
L(t)=
\begin{cases}
 \frac{1}{\beta_{L}}[e^{\beta_{L}(T'-t)}-e^{\beta_{L}(T-t)}]a(t), &\beta_L \neq 0,\medskip\\
 (T'-T)a(t), &\beta_L=0, 
\end{cases}
\end{align*}
where \[\beta_{L}=\mu_a-r_{0}-\sigma_a\rho_{Ia}\xi_1-\sigma_a\sqrt{1-\rho_{Ia}^2}\xi_2.\] 
In particular, we set \[\ell_0:=L(0)=\frac{1}{\beta_{L}}[e^{\beta_{L}T'}-e^{\beta_{L} T}]a_0.\]

\item The process $Z(t)$ satisfies 
\begin{align}\label{dynamicz}
\dd Z(t)=r_{0} Z(t)\dt+(\sigma^{\top}\pi(t)+\sigma_D D(t)-\sigma_L L(t))^{\top} (\xi \dt+\dd W(t)),
\end{align}
where 
\[\sigma_D=\big(\sigma_Y\rho_{IY},\sigma_Y\sqrt{1-\rho_{IY}^2}\big)^{\top}, \quad \sigma_L=\big(\sigma_a\rho_{Ia},\sigma_a\sqrt{1-\rho_{Ia}^2}\big)^{\top}.\]
Moreover, $\rho(t)Z(t)$ is a local martingale. 

\end{enumerate}

\end{proposition}

\begin{proof}
Note $\rho$ and $Y$ are time homogeneous processes, so 
\[D(t)=c\be_t\Big[\int_t^T\frac{\rho(s)}{\rho(t)}Y(s)\ds\Big]
=cY(t)\be_t\Big[\int_t^T\frac{\rho(s)}{\rho(t)}\frac{Y(s)}{Y(t)}\ds\Big]
=cY(t)\be\Big[\int_0^{T-t}\frac{\rho(s)}{\rho(0)}\frac{Y(s)}{Y(0)}\ds\Big].\]
Clearly
\[\frac{\rho(s)}{\rho(0)}\frac{Y(s)}{Y(0)}=e^{\beta_{D} s+(\sigma_D-\xi)^{\top}W(s)-\frac{1}{2}\Vert\sigma_D-\xi\Vert^2s},\] 
so it yields 
\[D(t)=\frac{1}{\beta_{D}}(e^{\beta_{D}(T-t)}-1)cY(t).\] 
The second claim can be proved similarly. The last claim is a consequence of It\^{o}'s lemma. 
\end{proof}

Since $\rho(t)Z(t)$ is a local martingale and $Z(T)\geq \ell$, we get 
\[\be[\rho(T)Z(T)]\leq \rho(0)Z(0)=z_0,\]
where \[z_0:=x_0+d_0-\ell_0.\] 
The above budget constraint restricts the choice of $Z(T)$. 

This result inspires us to consider the following static stochastic optimization problem over random variables $Z(T)$ subject to a budget constraint, a \ltvar constraint and a PI constraint: 
\begin{equation}\label{e3}
\begin{split}
\underset{Z(T)\in L^0_{\mathcal{F}_T} }{\sup}\quad &\be[U(Z(T))]\\
\mathrm{s.t.} \quad &\be[\rho Z(T)]\leq z_0\\
&\tvar^{-}_{\alpha}(Z(T))\geq \esb, \\
&Z(T)\geq\ell.
\end{split}
\end{equation} 
where $L^0_{\mathcal{F}_T}$ denotes all $\mathcal{F}_T$-measurable random variables. 
Here and hereafter, we write $\rho(T)$ as $\rho$ for notation simplicity if no confusion would occur.

Clearly, the optimal value of \eqref{e3} provides an upper bound for that of \eqref{e2}. Let us show they are indeed the same if \eqref{e3} admits an optimal solution. 

Suppose $Z^*(T)$ is an optimal solution to the problem \eqref{e3}. Because the objective in \eqref{e3} is increasing in $Z^*(T)$, we must have $\be[\rho(T) Z^*(T)]=z_0$. 
Let
\[Z^*(t)=\frac{\be_t[\rho(T) Z^*(T)]}{\rho(t)}.\]
Then $\rho(t)Z^*(t)=\be_t[\rho(T) Z^*(T)]$ is a martingale, so by the martingale representation theorem, 
there exists an $\mathcal{F}$-adapted progressively measurable process $v$ such that 
$ \int_0^T\Vert v(t)\Vert^2\dt<\infty$
and 
\[\rho(t) Z^*(t)=\be[\rho(T) Z^*(T)]+\int_0^t v(s)^{\top}\dd W(s)=z_0+\int_0^t v(s)^{\top}\dd W(s).\] 
By this equation and It\^{o}'s lemma, one can show that $Z^*(t)$ is a diffusion process, given by 
\[\dd Z^*(t)=r_{0} Z^*(t)\dt+[v(t)^{\top}\frac{1}{\rho(t)}+\xi^{\top}Z^*(t)](\xi dt+dW(t)),\]
Let
\[X^*(t)=Z^*(t)+L(t)-D(t),\]
and
\[ \pi^*(t)=\frac{1}{\rho(t)}(\sigma^{-1})^{\top}v(t)+Z^*(t)(\sigma^{-1})^{\top}\xi-(\sigma^{-1})^{\top}\sigma_D D(t)+(\sigma^{-1})^{\top}\sigma_L L(t) .\]
Then by It\^{o}'s lemma, one can verify that $(X^*(t),\pi^*(t))$ satisfies \eqref{e1}. Also, clearly $X^*(0)=x_0$ and $X^*(T)=Z^*(T)+L(T)$. 
Under the control $\pi^*(t)$, the objective value $\be[U(X^*(T)-L(T))]$ in the problem \eqref{e2} is equal to $\be[U(Z^*(T))]$, the optimal value of the problem \eqref{e3}. 
But the optimal value of \eqref{e3} is an upper bound for that of \eqref{e2}, so $\pi^*(t)$ is an optimal solution to the original problem \eqref{e2}.

Similarly, if $\pi^*(\cdot)$ is an optimal solution to the problem \eqref{e2}, then $Z^*(T)=X^*(T)- L(T)$ must be an optimal solution to the problem \eqref{e3}. 
Therefore, the problems \eqref{e2} and \eqref{e3} either both of them admit optimal solutions, or none of them.

\begin{remark}
For the upper bounded \tvar constraint model, the \ltvar constraint 
in \eqref{e3} shall be replaced by 
\[\tvar^{-}_{\alpha}(Z(T))\leq \esb.\]
\end{remark}
\begin{remark} 
\begin{enumerate}
\item If replacing the \ltvar constraint with VaR, the problem becomes
\begin{equation*}
\begin{split}
\underset{Z(T)\in L^0_{\mathcal{F}_T} }{\sup}\quad &\be[U(Z(T))]\\
\mathrm{s.t.}\quad &\be[\rho Z(T)]\leq z_0\\
& \mathbb{P}(Z(T)\geq \esb)\geq 1-\alpha,\\
&Z(T)\geq\ell.
\end{split}
\end{equation*}
This has been studied in Chen et al. \cite{C18}. 

\item If considering both the initial-time and intermediate-time VaR constraints, the problem is
\begin{equation*}
\begin{split}
\underset{Z(T)\in L^0_{\mathcal{F}_T} }{\sup}\quad &\be[U(Z(T))]\\
\mathrm{s.t.}\quad &\be[\rho Z(T)]\leq z_0\\
& \mathbb{P}(Z(T)\geq \esb)\geq 1-\alpha_0\\
& \mathbb{P}(Z(T)\geq \esb|\mathcal{F}_s)\geq 1-\alpha_1,\\
&Z(T)\geq\ell.
\end{split}
\end{equation*}
This has been studied in Wu et al. \cite{W22}.
\end{enumerate}

\end{remark}

\begin{remark}
In our model, $\rho(T)$ follows a log normal distribution. 
In order to solve the static problem \eqref{e3}, however, we do not need such a particular property. Indeed, our method only needs the following properties for $\rho(T)$: 1. $\be[\rho(T)]<\infty$; 2. its cumulative distribution function $F_{\rho(T)}$ is continuous, strictly increasing on $[0,\infty)$, and satisfies $F_{\rho(T)}(0)=0$. 
These properties can be satisfied for some models with stochastic parameters. 
Using the idea of \cite{X14}, the second requirement can be removed as well. 
\end{remark}

We call an optimization problem feasible if there exists a candidate (called a feasible solution) such that it satisfies all the constraints in the problem. Clearly, it is not interesting to investigate an infeasible (i.e. not feasible) problem, because there is nothing to choose. 
We call an optimization problem well-posed if it admits a finite optimal value (of course, it should be feasible first). 
If a problem is ill-posed (i.e. not well-posed), then one can achieve arbitrary large optimal value, so it is not interesting as well. For a well-posed optimization problem, if one feasible solution achieves its optimal value, then it is called an optimal solution to the problem. 

Before solving the problem \eqref{e3}, we need to study its feasibility and well-posedness issues. To study the latter, we first introduce a benchmark problem.

\subsection{A benchmark problem} 
To study the static problem \eqref{e3}, one nature benchmark problem is as follows 
\begin{equation}\label{benchmark}
\begin{split}
\underset{Z(T)\in L^0_{\mathcal{F}_T} }{\sup}\quad &\be[U(Z(T))]\\
\mathrm{s.t.} \quad &\be[\rho Z(T)]\leq z_0,\\
&Z(T)\geq\ell.
\end{split}
\end{equation} 
This is the problem \eqref{e3} without the \ltvar constraint. 

Throughout the paper, we let
\[\barz:=z_0-\ell\be[\rho].\]
If the problem \eqref{benchmark} admits a feasible solution $Z(T)$, then 
\[\barz\geq \be[\rho Z(T)]-\ell\be[\rho]
=\be[( Z(T)-\ell)\rho]\geq 0.\]
Therefore, if $\barz<0$, then the problem \eqref{benchmark} is infeasible, so is the problem \eqref{e3}. Also, if $\barz=0$, then all the inequalities in above shall be equations, thus, $Z(T)=\ell$ is the unique feasible (thus optimal) solution to the problem \eqref{benchmark} with the finite optimal value $U(\ell)$. 
Meanwhile, when $\ell\geq \esb$, $Z(T)=\ell$ is the unique feasible (thus optimal) solution to the problem \eqref{e3} with the finite optimal value $U(\ell)$; whereas when $\ell<\esb$, there is no feasible solution to the problem \eqref{e3}, so the problem \eqref{e3} is infeasible. 

It is only left to study the case $\barz>0$. The following requirement is minimum and henceforth assumed. 
\begin{assumption}\label{ass1}
The benchmark problem \eqref{benchmark} is well-posed when $\barz>0$. 
\end{assumption}
Assumption \ref{ass1} is abstract and hard to verify. We now express it by some explicitly and easily verified conditions in terms of the parameters. 
In particular, the result shows that if the benchmark problem \eqref{benchmark} is well-defined for some $\barz>0$, then so is for any other $\barz>0$.
\begin{lemma}
Assumption \ref{ass1} holds if and only if there exists some $\lam >0$ such that 
\begin{align}\label{wellposecondition}
\be[\rho (U')^{-1}(\lam \rho)]<\infty, \quad \be\big[ U\big((U')^{-1}(\lam \rho)\big)\big]<\infty.
\end{align}
\end{lemma}
\begin{proof}
By the change of variable, $Y=Z(T)-\ell$, we see the problem \eqref{benchmark} is well-posed if and only if so is the following problem 
\begin{equation*}
\underset{Y\in L^0_{\mathcal{F}_T} }{\sup}\;\; \be[U(Y+\ell)]\quad \mathrm{s.t.} \quad \be[\rho Y]\leq \barz,\quad Y\geq 0.
\end{equation*} 
Because $U$ is increasing and concave, we have 
\[U(x)\leq U(x+\ell)\leq U(x)+U(\ell)-U(0),\quad x\geq 0.\]
Hence the above problem is well-posed if and only if so is the following problem 
\begin{equation}\label{benchmark1}
\underset{Y\in L^0_{\mathcal{F}_T} }{\sup}\;\; \be[U(Y)]\quad \mathrm{s.t.} \quad \be[\rho Y]\leq \barz,\quad Y\geq 0.
\end{equation} 
Therefore, the well-posedness of \eqref{benchmark} is equivalent to that of \eqref{benchmark1}.

Suppose Assumption \ref{ass1} holds. Then \eqref{benchmark1} is well-posed. 
By \cite[Theorem 3.1]{JXZ08}, 
there exists some $\lam >0$ such that 
$\be[\rho (U')^{-1}(\lam \rho)]<\infty$.
By \cite[Lemma 3.1]{JXZ08}, the following problem is also well-posed, 
\begin{equation}\label{benchmark2}
\underset{Y\in L^0_{\mathcal{F}_T} }{\sup}\;\; \be[U(Y)]\quad \mathrm{s.t.} \quad \be[\rho Y]\leq \be[\rho (U')^{-1}(\lam \rho)],\quad Y\geq 0.
\end{equation} 
Since $(U')^{-1}(\lam \rho)$ is a feasible solution to the above problem, 
its value must be finite, i.e., 
\begin{align*} 
\be\big[ U\big((U')^{-1}(\lam \rho)\big)\big]<\infty.
\end{align*}
Hence, \eqref{wellposecondition} holds.

To show the reverse implication, suppose 
there exists some $\lam >0$ such that \eqref{wellposecondition} holds. 
Then by \cite[Theorem 1.1]{JXZ08}, $(U')^{-1}(\lam \rho)$ is an optimal solution to the problem \eqref{benchmark2} with the optimal value 
\begin{align*} 
\be\big[ U\big((U')^{-1}(\lam \rho)\big)\big]<\infty.
\end{align*}
In view of \cite[Lemma 3.1]{JXZ08}, the problem \eqref{benchmark1} is also well-posed. But the latter is equivalent to Assumption \ref{ass1} holds, so the proof is complete. 
\end{proof}

\subsection{Quantile formulation} 
We will use the quantile method to tackle the problem \eqref{e3}. This method turns the optimal maximization problem over random variables into an optimal maximization problem over their quantile functions.

The quantile function $G$ of a random variable $X$ is defined as 
\[G(z)={\rm VaR}^{-}_{z}(X), \quad z\in(0,1),\]
with the convention that $G(0)=\lim_{z\to 0^{+}} G(z)$ and $G(1)=\lim_{z\to 1-} G(z)$.
It is not hard to verify that quantiles are those left-continuous and increasing functions on $(0,1)$, and vise versa.

Thanks to the strictly monotonicity of $U$, applying \cite[Theorem 9]{X14}, we obtain 
\begin{proposition}\label{optimalz}
A random variable $Z^*(T)$ is an optimal solution to the problem \eqref{e3} if and only if 
it can be expressed as 
\[Z^*(T)=\barg(1-F_{\rho}(\rho))+\ell\]
where $\barg$ is an optimal solution to the following quantile optimization problem 
\begin{equation}\label{e4}
\begin{split}
\underset{G\in\setg}{\sup}\quad&\int_{0}^{1}U(G(z)+\ell)\dz\\
\mathrm{s.t.} \quad &\int_{0}^{1}G(z)F_{\rho}^{-1}(1-z)\dz\leq \barz,\\
&\frac{1}{\alpha}\int_{0}^{\alpha}G(z)\dz\geq \underz, 
\end{split}.
\end{equation}
and
\[ \underz:=\max\{\esb-\ell,0\},\]
and $F_{\rho}^{-1}(\cdot)$ is the quantile function of $\rho$, and $\setg$ is the set of quantiles for nonnegative random variables, given by \[\setg=\Big\{G(\cdot): (0,1)\to [0,\infty), \text{ left-continuous and increasing on }(0,1) \Big\}.\]
\end{proposition} 

By this result, it is sufficient to study the quantile optimization problem \eqref{e4}.

\begin{remark}
In above we replaced $\esb-\ell$ by $\max\{\esb-\ell,0\}$ since $G\geq 0$. 
\end{remark} 
\begin{remark}
One is interested in the upper bounded \tvar constraint model. Then the \ltvar constraint 
in \eqref{e4} shall be replaced by 
\[\frac{1}{\alpha}\int_{0}^{\alpha}G(z)\dz\leq \underz.\]
\end{remark}

\subsection{Feasibility and well-posedness of \eqref{e4}}

By Proposition \ref{optimalz}, the optimal value of the problem \eqref{e4} is dominated by that of \eqref{benchmark}. 
When it is feasible, the problem \eqref{e4} is well-posed 
since the problem \eqref{benchmark} is well-posed under Assumption \ref{ass1}. 
So our problem turns to study the feasibility issue of the problem \eqref{e4}. 
This issue will be resolved via solving the following optimization problem 
\begin{equation}\label{e5}
\begin{split}
\underset{G\in\mathcal{G}}{\inf}\quad &\int_{0}^{1}G(z)F_{\rho}^{-1}(1-z)\dz\\
\text{s.t.} \quad &\frac{1}{\alpha}\int_{0}^{\alpha}G(z)\dz\geq \underz. 
\end{split} 
\end{equation} 
To solve this problem, we need a technical result. 

\begin{lemma}\label{lamup}
The function 
\begin{align}\label{functioneta}
\eta(\lam):=-\int_{1-F_{\rho}\big(\tfrac{\lam}{\alpha}\big)}^1F_{\rho}^{-1}(1-y)\dy+\lam \frac{\alpha-1+F_{\rho}\big(\tfrac{\lam}{\alpha}\big)}{\alpha}, \quad \lam\in (0,\infty),
\end{align} 
admits a unique root $\lamup$, which satisfies $\alpha F_{\rho}^{-1}(1-\alpha)<\lamup< \be[\rho]$. The function $\eta$ is strictly decreasing and negative on $(0,\alpha F_{\rho}^{-1}(1-\alpha)]$ and strictly increasing on $[\alpha F_{\rho}^{-1}(1-\alpha),\infty)$. 
Moreover, $\lamup$ is the unique minimizer for 
\[\zeta(\lam)=\frac{\int_{1-F_{\rho}(\tfrac{\lam}{\alpha})}^1F_{\rho}^{-1}(1-z)\dz}{\alpha-1+F_{\rho}(\tfrac{\lam}{\alpha})},\quad \lam\in (\alpha F_{\rho}^{-1}(1-\alpha),\infty],\]
with the minimum $\zeta(\lamup)=\frac{\lamup}{\alpha}$. 
\end{lemma}
\begin{proof}
Direct calculation yields
\begin{align*} 
\eta(0+)=0,\quad \eta(\alpha F_{\rho}^{-1}(1-\alpha))=-\int_{\alpha}^1F_{\rho}^{-1}(1-y)\dy<0,\quad \lim_{\lam\to+\infty}\eta(\lam)=+\infty,
\end{align*} 
and
\begin{align} 
\eta'(\lam)= \frac{\alpha-1+F_{\rho}\big(\tfrac{\lam}{\alpha}\big)}{\alpha}
=\begin{cases}
<0,&\quad \lam< \alpha F_{\rho}^{-1}(1-\alpha);\\
>0,&\quad \lam> \alpha F_{\rho}^{-1}(1-\alpha).
\end{cases}
\end{align}
So $\eta$ is strictly decreasing and negative on $(0,\alpha F_{\rho}^{-1}(1-\alpha)]$ and strictly increasing on $[\alpha F_{\rho}^{-1}(1-\alpha),\infty)$. 
Also, $\eta$ admits a unique root $\lamup$ on $(0,\infty)$, which satisfies $\lamup>\alpha F_{\rho}^{-1}(1-\alpha)$. 

We now prove $\lamup<\be[\rho]$. 
When $z<1-F_{\rho}(\tfrac{\lam}{\alpha})$, we have 
\begin{align*} 
\frac{\partial}{\partial z}\Big(-\int_{z}^1F_{\rho}^{-1}(1-y)\dy+\lam \frac{\alpha-z}{\alpha}\Big)
=F_{\rho}^{-1}(1-z)-\frac{\lam}{\alpha}>0,
\end{align*} 
so
\begin{align*} 
\eta(\lam)&=\Big(-\int_{z}^1F_{\rho}^{-1}(1-y)\dy+\lam \frac{\alpha-z}{\alpha}\Big)\Big|_{z=1-F_{\rho}(\tfrac{\lam}{\alpha})}\\
&>\Big(-\int_{z}^1F_{\rho}^{-1}(1-y)\dy+\lam \frac{\alpha-z}{\alpha}\Big)\Big|_{z=0} =\lam-\be[\rho].
\end{align*} 
Since $\eta(\lamup)=0$, if follows $\lamup<\be[\rho]$. 

We now focus on $\zeta$. 
For $\lam>\alpha F_{\rho}^{-1}(1-\alpha)$, we have 
\[\big(\alpha-1+F_{\rho}(\tfrac{\lam}{\alpha})\big)\zeta(\lam)=\int_{1-F_{\rho}(\tfrac{\lam}{\alpha})}^1F_{\rho}^{-1}(1-z)\dz,\]
so differentiating both sides gives 
\begin{align*}
\big(\alpha-1+F_{\rho}(\tfrac{\lam}{\alpha})\big)\zeta'(\lam)
+F'_{\rho}(\tfrac{\lam}{\alpha})\frac{1}{\alpha}\zeta(\lam)
=\frac{\lam}{\alpha}F'_{\rho}(\tfrac{\lam}{\alpha})\frac{1}{\alpha}.
\end{align*} 
It thus follows that %
\begin{align*}
\zeta'(\lam) &=\frac{\frac{\lam}{\alpha}-\zeta(\lam)}{\alpha-1+F_{\rho}(\tfrac{\lam}{\alpha})} F'_{\rho}(\tfrac{\lam}{\alpha})\frac{1}{\alpha}\\
&=\frac{\eta(\lam)}{\big(\alpha-1+F_{\rho}(\tfrac{\lam}{\alpha})\big)^2} F'_{\rho}(\tfrac{\lam}{\alpha})\frac{1}{\alpha}
\begin{cases}
<0,&\; \lam\in (\alpha F_{\rho}^{-1}(1-\alpha),\lamup);\\
>0,&\; \lam\in (\lamup,\infty).
\end{cases}
\end{align*} 
Hence, we conclude $\lamup$ is the unique minimizer for $\zeta$. 
Since $\eta(\lamup)=0$, it also gives $\zeta(\lamup)=\frac{\lamup}{\alpha}$.
\end{proof}

\begin{lemma}\label{feasible1}
Let $\lamup$ be given in Lemma \ref{lamup}.
Then the unique optimal solution to the problem \eqref{e5} is given by 
\begin{align}\label{gstar}
G^{*}(z)=\begin{cases}
0, &\quad z\in (0,1-F_{\rho}(\frac{\lamup}{\alpha})];\\
\frac{\alpha }{\alpha-1+F_{\rho}(\frac{\lamup}{\alpha})}\underline{z}, &\quad z\in (1-F_{\rho}(\frac{\lamup}{\alpha}),1),
\end{cases} 
\end{align} 
with the optimal value 
\begin{align}\label{functionR}
R(\underz) 
=\lamup\underz.
\end{align} 
Moreover, the problem \eqref{e4} is feasible if and only if $\barz\geq R(\underz)$. Especially when $\barz=R(\underz)$, the solution of problem \eqref{e4} is uniquely given as $G^{*}$ above. 
\end{lemma}
\begin{proof}
If $\underz=0$, then clearly $G^*\equiv 0$ is the unique optimal solution to \eqref{e5} by the non-negativity of quantiles. 

Now suppose $\underz>0$. Then $G^*$ defined in \eqref{gstar} is feasible solution to the problem \eqref{e5}.
Let $G$ be any feasible solution to the problem \eqref{e5}.

We first suppose $G$ is not a constant on $(0,\alpha)$. Then 
\[\frac{1}{\alpha}\int_{0}^{\alpha}G(\alpha)\dz
>\frac{1}{\alpha}\int_{0}^{\alpha}G(z)\dz\geq \underz.\]
Hence, there exists $\lam\in (\alpha F^{-1}_{\rho}(1-\alpha),\infty)$ such that 
\begin{align}\label{lam}
\frac{1}{\alpha}\int_{1-F_{\rho}(\tfrac{\lam}{\alpha})}^{\alpha}G(\alpha)\dz=\underz.
\end{align}
Thus,
\[G(\alpha)=\frac{\alpha \underz}{\alpha-1+F_{\rho}(\tfrac{\lam}{\alpha})}.\]
By the monotonicity and non-negativity of quantiles, we have 
\begin{align*}
\int_{0}^{1}G(z)F_{\rho}^{-1}(1-z) \dz
&\geq \int_{0}^{1}G(z)F_{\rho}^{-1}(1-z)+\lam \Big(\underz-\frac{1}{\alpha}\int_{0}^{\alpha}G(z)\dz\Big)\\
&=\int_{0}^{\alpha}G(z)\Big(F_{\rho}^{-1}(1-z)-\frac{\lam}{\alpha}\Big)\dz+\int_{\alpha}^1G(z)F_{\rho}^{-1}(1-z)\dz+\lam \underz\\ 
&\geq \int_{0}^{1-F_{\rho}(\tfrac{\lam}{\alpha})} 0\cdot\Big(F_{\rho}^{-1}(1-z)-\frac{\lam}{\alpha}\Big)\dz\\
&\quad\;+\int_{1-F_{\rho}(\tfrac{\lam}{\alpha})}^{\alpha}G(\alpha)(F_{\rho}^{-1}(1-z)-\frac{\lam}{\alpha})\dz\\
&\quad\;+\int_{\alpha}^1G(\alpha)F_{\rho}^{-1}(1-z)\dz+\lam \underz\\
&=
\int_{1-F_{\rho}(\tfrac{\lam}{\alpha})}^{1}\frac{\alpha \underz}{\alpha-1+F_{\rho}(\tfrac{\lam}{\alpha})} F_{\rho}^{-1}(1-z)\dz\\
&=\alpha \underz \zeta(\lam)\\
&\geq \alpha \underz \zeta(\lamup)=\int_{0}^{1}G^{*}(z)F_{\rho}^{-1}(1-z)\dz, 
\end{align*}
where $\zeta$ is defined in Lemma \ref{lamup} and $G^{*}$ is given by \eqref{gstar}. 
Since $\alpha\underz>0$ and $\lamup$ is the unique minimizer for $\zeta$, the above inequalities become equations if and only if $G\equiv G^*$.

On the other hand, if $G$ is a constant on $(0,\alpha)$, then 
\[G(0)=\frac{1}{\alpha}\int_{0}^{\alpha}G(\alpha)\dz\geq \underz,\]
and thus,
\begin{align*}
\int_{0}^{1}G(z)F_{\rho}^{-1}(1-z) 
&\geq \underz \int_{0}^{1} F_{\rho}^{-1}(1-z)= 
\alpha \underz \zeta(\infty)>\alpha \underz \zeta(\lamup)=\int_{0}^{1}G^{*}(z)F_{\rho}^{-1}(1-z)\dz.
\end{align*}
This shows $G^*$ is the unique optimal solution to \eqref{e5} with the optimal value $\alpha \underz \zeta(\lamup)=\lamup \underz$ by virtue of Lemma \ref{lamup}. 

The other claims are trivial. 
\end{proof}

\begin{remark}
For the upper bounded \tvar constraint model, corresponding to \eqref{e4}, the problem is feasible if and only if $\barz\geq0$, $\underz\geq 0$. 
In particular, the optimal solution is $G\equiv 0$ if $\barz=0$ or $\underz=0$. 
the problem is well-posed if $\barz>0$, $\underz>0$ and Assumption \ref{ass1} holds.
\end{remark}

\section{Quantile optimization and optimal solution}\label{sec:quantile}

In this section, we solve the original problem \eqref{e2} by quantile optimization techniques.

We start with the quantile optimization problem \eqref{e4}. 
There are four cases. 
\subsection{No feasible solution: $\barz<R(\underz)$.}
If $\barz<R(\underz)$, then by Lemma \ref{feasible1}, the problem \eqref{e4} is infeasible, i.e., it admits no feasible solution.

\subsection{A unique solution: $\barz=R(\underz)$.}
If $\barz=R(\underz)$, then by Lemma \ref{feasible1}, the problem \eqref{e4} admits a unique feasible, thus optimal solution given by \eqref{gstar}.

\subsection{Ineffective \ltvar constraint: $\barz>R(\underz)$, $\barz\geq C(\underz)$.}
When $\barz>R(\underz)$, the problem \eqref{e4} is feasible and well-defined under Assumption \ref{ass1}. 

The budget constraint in the problem \eqref{e4} must be effective (namely it holds with equality) for the optimal solution, if it exists, since a larger quantile is always preferred for both the constraint set and the objective functional. By contrast, the \ltvar constraint may not be effective. 
We now give the answer to the case with an ineffective \ltvar constraint. 

Let
\begin{align}\label{functionC}
C(\underz)=\int_{0}^{1}\big((U')^{-1}(\underlam F_{\rho}^{-1}(1-z))-\ell\big)^+F_{\rho}^{-1}(1-z)\dz,
\end{align}
where $\underlam=\underlam(\underz) >0$ for $\underz>0$ is uniquely determined by 
\begin{align}\label{functionC2}
\frac{1}{\alpha}\int_{0}^{\alpha}\big((U')^{-1}(\underlam F_{\rho}^{-1}(1-z))-\ell\big)^+\dz=\underz.
\end{align}

\begin{theorem}
[Optimal solution for the problem \eqref{e4} with an ineffective \ltvar constraint] 
\label{thm:ineffictive}
Suppose there exists $\lam>0$ such that 
\[\int_{0}^{1}\big((U')^{-1}(\lam F_{\rho}^{-1}(1-z))-\ell\big)^+ F_{\rho}^{-1}(1-z)\dz=\barz.\]
Then \[\barg(z)=\big((U')^{-1}(\lam F_{\rho}^{-1}(1-z))-\ell\big)^+\]
is optimal to the problem \eqref{e4} if and only if 
\[\barz\geq C(\underz),\]
or equivalently 
\[\underlam(\underz)\geq \lam.\]
\end{theorem}

\begin{proof}
For any $G\in\setg$ such that 
\[\int_{0}^{1}G(z)F_{\rho}^{-1}(1-z)\dz\leq \barz,\]
we have
\begin{align*}
\int_{0}^{1}U(G(z)+\ell)\dz &\leq \int_{0}^{1}\big(U(G(z)+\ell)-\lam G(z)F_{\rho}^{-1}(1-z)\big)\dz+\lam\underz\\
&\leq \int_{0}^{1}\big(U(\barg(z)+\ell)-\lam \barg(z)F_{\rho}^{-1}(1-z)\big)\dz+\lam\underz\\
&= \int_{0}^{1}\big(U(\barg(z)+\ell)\dz,
\end{align*}
where $\lam$ and $\barg$ are given in the hypothesis, and we used the fact that 
$ \barg(z)$ maximizes the mapping
\[x\mapsto U(x+\ell)-\lam xF_{\rho}^{-1}(1-z),\quad x\geq 0.\] 
By hypothesis, we have $\barg\in\setg$, so it is an optimal solution to the following problem without \ltvar constraint: 
\begin{equation}\label{e7}
\begin{split}
\underset{G\in\setg}{\sup}\quad &\int_{0}^{1}U(G(z)+\ell)\dz\\
\mathrm{s.t.} \quad &\int_{0}^{1}G(z)F_{\rho}^{-1}(1-z)\dz\leq \barz.
\end{split}
\end{equation}
So $\barg$ is optimal to the problem \eqref{e4} if and only if it satisfies the \ltvar constraint. 
It is effortless to obtain the result by the strict monotonicity of $U'(\cdot)$.
\end{proof}

\subsection{Effective \ltvar constraint: $R(\underz)<\barz<C(\underz)$.}

Next we deal with the most interesting case: $R(\underz)<\barz<C(\underz)$. In this case, 
the budget constraint and the \ltvar constraint in \eqref{e4} are both effective. 

We now apply the Lagrange dual approach to study the problem \eqref{e4}. 
Notice the \ltvar constraint can be rewritten as 
\[\frac{1}{\alpha}\int_0^{1}G(z)I_{\{z\leq\alpha\}}\dz\geq \underz.\] 
So we consider the following auxiliary problem for each pair Lagrange multipliers $(\lam,\lamb)\in (0,\infty)\times(0,\infty)$:
\begin{equation}\label{e8a}
\begin{split} 
\underset{G\in\setg}{\sup}\;&\int_0^1U(G(z)+\ell)\dz-\lamb\Big(\int_{0}^{1}G(z)F_{\rho}^{-1}(1-z)\dz-\barz\Big)+\lamb\lam \Big(\frac{1}{\alpha}\int_0^{1}G(z)I_{\{z\leq \alpha\}}\dz-\underz\Big)\\
=\underset{G\in\setg}{\sup}\;&\int_0^1\big(U(G(z)+\ell)-\lamb G(z)\varphi'_{\lam}(z)\big)\dz+\lamb \barz-\lam\lamb\underz,
\end{split}
\end{equation}
where 
\[\varphi_{\lam}(z):=-\int_z^1F_{\rho}^{-1}(1-y)\dy+\lam \frac{\alpha-z}{\alpha} I_{\{z\leq \alpha\}},\] 
whose left derivative is given by 
\[\varphi'_{\lam}(z):=F_{\rho}^{-1}(1-z)-\frac{\lam}{\alpha} I_{\{z\leq \alpha\}}.\] 
Because $U$ is strictly concave, the problem \eqref{e8a} admits at most one solution.

The relationship between \eqref{e8a} and \eqref{e4} is contained in the following result. 
\begin{lemma}\label{t1}
Let $G^{*}_{\lam,\lamb}$ be the optimal solution to the problem \eqref{e8a}. If there exists a pair of Lagrange multiplier $(\lam^*,\lamb^*)\in (0,\infty)\times(0,\infty)$ such that 
\[\int_0^1 G^{*}_{\lam^{*},\lamb^{*}}(z)F_{\rho}^{-1}(1-z)\dz=\barz,\quad 
\frac{1}{\alpha}\int_{0}^{\alpha}G^{*}_{\lam^{*},\lamb^{*}}(z)\dz=\underz,\]
then $G^{*}_{\lam^{*},\lamb^{*}}$ is the optimal solution to the problem \eqref{e4}.
\end{lemma}
\begin{proof}
Clearly, $G^{*}_{\lam^{*},\lamb^{*}}$ is a feasible solution to problem \eqref{e4}. For any feasible solution $G\in\setg$, we have 
\begin{align*}
&\quad\;\int_0^1U(G^{*}_{\lam^{*},\lamb^{*}}(z)+\ell)\dz-\int_0^1U(G(z)+\ell)\dz\\
&\geq \int_0^1U(G^{*}_{\lam^{*},\lamb^{*}}(z)+\ell)\dz-\int_0^1U(G(z)+\ell)\dz\\
&\quad\;+\lamb^*\Big(\int_{0}^{1}G(z)F_{\rho}^{-1}(1-z)\dz-\int_0^1 G^{*}_{\lam^{*},\lamb^{*}}(z)F_{\rho}^{-1}(1-z)\dz\Big)\\
&\quad\;-\lamb^*\lam^* \Big(\frac{1}{\alpha}\int_0^{1}G(z)I_{\{z\leq \alpha\}}\dz-\frac{1}{\alpha}\int_{0}^{\alpha}G^{*}_{\lam^{*},\lamb^{*}}(z)\dz\Big)\\
&=\int_0^1\big(U(G^{*}_{\lam^{*},\lamb^{*}}(z)+\ell)-\lamb^{*}G^{*}_{\lam^{*},\lamb^{*}}(z)\varphi'_{\lam^{*}}(z)\big)\dz\\
&\quad\;-\int_0^1\big(U(G(z)+\ell)-\lamb^{*}G(z)\varphi'_{\lam^{*}}(z)\big)\dz\\
&\geq \ 0.
\end{align*}
where the last inequality is due to the optimality of $G^{*}_{\lam^{*},\lamb^{*}}$ to the problem \eqref{e8a}. 
This confirms the claim. 
\end{proof}
Based on Lemma \ref{t1}, we can disentangle the problem \eqref{e4} by firstly solving the problem \eqref{e8a} and then determining the Lagrange multipliers by the two constraints.

Before doing these, we first present several important technique results. 
\begin{lemma}\label{wellposed2}
Let $\lamup$ be given in Lemma \ref{lamup}.
The following inequality 
\begin{align}\label{monotone1}
\varphi_{\lam}(z)<\varphi_{\lam}(1)=0,\quad z\in[0,1),
\end{align}
holds true if and only if $0<\lam<\lamup$. Moreover, 
the problem \eqref{e8a} is ill-posed if $(\lam,\lamb)\in [\lamup, \infty)\times(0,\infty)$. 
\end{lemma}
\begin{proof} 
Trivially, 
\[ \varphi_{\lam}(z)=-\int_{z}^1F_{\rho}^{-1}(1-y)\dy<0=\varphi_{\lam}(1),\quad z\in (\alpha,1), \] 
so \eqref{monotone1} holds true if and only if
\[\max_{z\in [0,\alpha]}\varphi_{\lam}(z)<0.\]
 When $z\leq \alpha$, 
\[\varphi'_{\lam}(z)=F_{\rho}^{-1}(1-z)- \frac{\lam}{\alpha}
=\begin{cases}
>0, &\quad z< 1-F_{\rho}\big(\tfrac{\lam}{\alpha}\big);\\
<0, &\quad z> 1-F_{\rho}\big(\tfrac{\lam}{\alpha}\big).
\end{cases}
\]
Hence 
\[\max_{z\in [0,\alpha]}\varphi_{\lam}(z)
=\begin{cases}
\varphi_{\lam}(\alpha), &\quad \alpha\leq 1-F_{\rho}\big(\tfrac{\lam}{\alpha}\big);\\
\varphi_{\lam}\big(1-F_{\rho}\big(\tfrac{\lam}{\alpha}\big)\big), &\quad \alpha> 1-F_{\rho}\big(\tfrac{\lam}{\alpha}\big).
\end{cases}
 \]
There are two cases. 
\begin{itemize} 

\item If $0<\lam\leq \alpha F_{\rho}^{-1}(1-\alpha)$, then $1-F_{\rho}\big(\tfrac{\lam}{\alpha}\big)\geq \alpha$. Thus 
\[\max_{z\in [0,\alpha]}\varphi_{\lam}(z)=\varphi_{\lam}(\alpha)
=-\int_{\alpha}^1F_{\rho}^{-1}(1-y)\dy<0.
\] 
\item If $\lam>\alpha F_{\rho}^{-1}(1-\alpha)$, then $1-F_{\rho}\big(\tfrac{\lam}{\alpha}\big)<\alpha$. And consequently,
\begin{align*}
\max_{z\in [0,\alpha]}\varphi_{\lam}(z)&=\varphi_{\lam}\big(1-F_{\rho}\big(\tfrac{\lam}{\alpha}\big)\big)\\
&=-\int_{1-F_{\rho}\big(\tfrac{\lam}{\alpha}\big)}^1F_{\rho}^{-1}(1-y)\dy+\lam \frac{\alpha-1+F_{\rho}\big(\tfrac{\lam}{\alpha}\big)}{\alpha}\\
&=\eta(\lam), 
\end{align*}
where $\eta$ is defined in Lemma \ref{lamup}. So by Lemma \ref{lamup}, $\max_{z\in [0,\alpha]}\varphi_{\lam}(z)<0$ happens if and only if $\lam<\lamup$. 
\end{itemize} 
Combining the above cases, we conclude that \eqref{monotone1} holds true if and only if $0<\lam<\lamup$. 

Now suppose $(\lam,\lamb)\in [\lamup, \infty)\times(0,\infty)$. 
Then there exists $z_0\in[0,1)$ such that $\varphi_{\lam}(z_0)\geq \varphi_{\lam}(1).$
Let $G(z)=kI_{\{z> z_0\}} $, then
\begin{align*}
\int_{0}^1\big(U(G(z)+\ell)-\lamb G(z)\varphi'_{\lam}(z)\big)\dz
&= U(\ell) z_0+ U(k+\ell)(1-z_0)-\lamb k (\varphi_{\lam}(1)-\varphi_{\lam}(z_0))\\
&\geq U(\ell) z_0+U(k+\ell)(1-z_0)\to\infty, \quad\mbox{as $k\to\infty$},
\end{align*}
so the problem \eqref{e8a} is ill-posed. 
The proof is complete. 
\end{proof}

It suffices to study the case $(\lam,\lamb)\in (0,\lamup)\times(0,\infty)$. 
\begin{lemma}\label{zonetwo}
Let $\lamup$ be given in Lemma \ref{lamup} and $0<\lam<\lamup $.
Then
there exists a pair $z_1=z_1(\lam)\in (0,\alpha)$ and $z_2=z_2(\lam)\in(\alpha,1)$ such that 
$z_1=s(z_2)$, where 
\[s(z)=1-F_{\rho}\Big(F_{\rho}^{-1}(1-z)+\frac{\lam}{\alpha}\Big),
\footnote{This indicates $1-F_{\rho}\Big(F_{\rho}^{-1}(1-\alpha)+\frac{\lam}{\alpha}\Big)< z_1< 1-F_{\rho}\big(\tfrac{\lam}{\alpha}\big)$ as $z_2\in(\alpha,1)$.
If $z_1$ is too small, then the tangent line is above $\varphi_{\lam}$ on $[\alpha,1]$, 
$\varphi_{\lam}(z_1)+\varphi'_{\lam}(z_1)(z-z_1)>\varphi_{\lam}(z)$, $z\in[\alpha,1]$. So we have some lower bound for $z_1$.
}\]
and $z_2$ is the unique solution on $(\alpha,1)$ to the following equation
\[\varphi_{\lam}(z)-\varphi_{\lam}(s(z))-\varphi'_{\lam}(z)(z-s(z))=0.\]
Moreover, the concave envelope of $\varphi_{\lam}$ on $[0, 1]$ coincides with $\varphi_{\lam}$ on $[0,z_1]\cup[z_2,1]$ and is linear on $[z_1, z_2]$, and
\[\varphi_{\lam}'(z_1)=\varphi_{\lam}'(z_2),\quad \lim_{\lam\to0}z_1(\lam)=\lim_{\lam\to0}z_2(\lam)=\alpha, \quad 
\lim_{\lam\to\lamup}z_1(\lam)=1-F_{\rho}\big(\tfrac{\lamup}{\alpha}\big),\quad \lim_{\lam\to\lamup}z_2(\lam)=1.\] 
\end{lemma}
\begin{figure}[htbp]
\centering
\includegraphics[height=8cm,width=10cm]{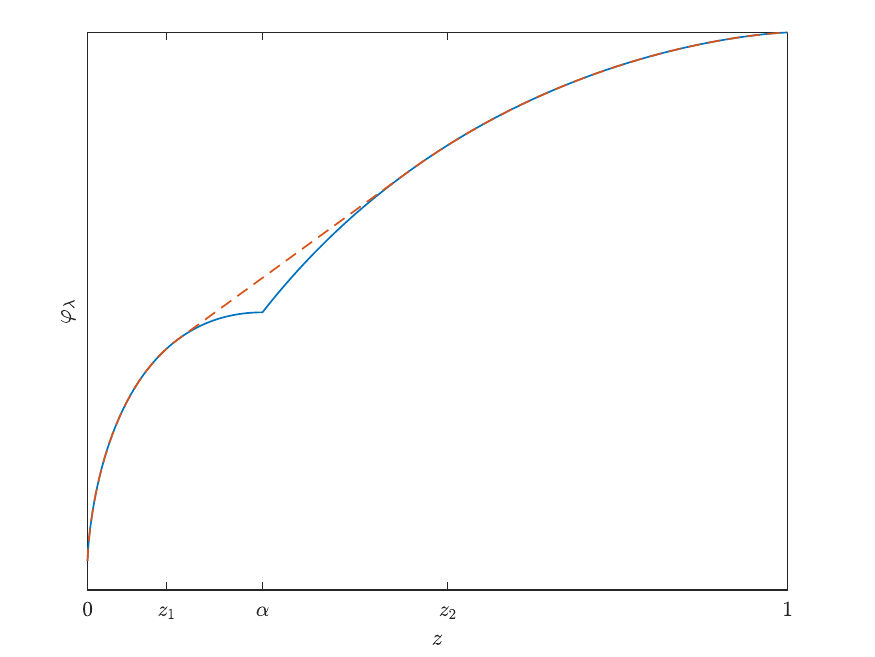}
\caption{$\varphi_{\lambda}$ and its concave envelope.}
\label{varphi}
\end{figure}
\begin{proof} 
When $0<\lam<\lamup $, by virtue of \eqref{monotone1}, we have 
\[\max_{z\in [0,\alpha]}\varphi_{\lam}(z)<0.\]
For any $x\in (\alpha,1)$, the tangent line of $\varphi_{\lam}$ at $x$ is 
\[z\mapsto \varphi_{\lam}(x)+\varphi'_{\lam}(x)(z-x)
=\varphi_{\lam}(x)+F_{\rho}^{-1}(1-x)(z-x).\]
Notice 
\begin{align*}
&\quad\min_{z\in [0,\alpha]}\big[\varphi_{\lam}(x)+F_{\rho}^{-1}(1-x)(z-x)\big]\\
&\geq \varphi_{\lam}(x)-F_{\rho}^{-1}(1-x)x\\
&=-\int_{x}^1F_{\rho}^{-1}(1-y)\dy-F_{\rho}^{-1}(1-x)x\to 0,\quad\mbox{as $x\to1$}.
\end{align*}
Let 
\begin{align*}
f(x)&=\min_{z\in [0,\alpha]}\big[\varphi_{\lam}(x)+F_{\rho}^{-1}(1-x)(z-x)-\varphi_{\lam}(z)\big].\end{align*}
Then, when $x$ is sufficiently close to 1,
\begin{align*}
f(x)&=\min_{z\in [0,\alpha]}\big[\varphi_{\lam}(x)+F_{\rho}^{-1}(1-x)(z-x)-\varphi_{\lam}(z)\big]\\
&\geq \min_{z\in [0,\alpha]}\big[\varphi_{\lam}(x)+F_{\rho}^{-1}(1-x)(z-x)\big]
-\max_{z\in [0,\alpha]}\varphi_{\lam}(z)>0. 
\end{align*} 
Geometrically speaking, the whole tangent line is above $\varphi_{\lam}$ on $[0,\alpha]$. 

On the other hand, for any $x\in (\alpha,1)$, $z\in [0,\alpha]$, 
\begin{align*} 
&\quad \varphi_{\lam}(x)+F_{\rho}^{-1}(1-x)(z-x)-\varphi_{\lam}(z)\\
&=-\int_{x}^1F_{\rho}^{-1}(1-y)\dy+F_{\rho}^{-1}(1-x)(z-x)
+\int_z^1F_{\rho}^{-1}(1-y)\dy-\lam \frac{\alpha-z}{\alpha}\\
&=\int_{z}^{x}F_{\rho}^{-1}(1-y)\dy+F_{\rho}^{-1}(1-x)(z-x)-\lam \frac{\alpha-z}{\alpha}\\
&\leq (x-z)F_{\rho}^{-1}(1-z)+F_{\rho}^{-1}(1-x)(z-x)-\lam \frac{\alpha-z}{\alpha}\\
&=(x-z)\Big(F_{\rho}^{-1}(1-z)-F_{\rho}^{-1}(1-x)\Big)-\lam \frac{\alpha-z}{\alpha}. 
\end{align*}
Let $x=\alpha+\ep^2$ and $z=\alpha-\ep$, then 
\begin{align*} 
&\quad\; \varphi_{\lam}(x)+F_{\rho}^{-1}(1-x)(z-x)-\varphi_{\lam}(z)\\ 
&\leq (x-z)\Big(F_{\rho}^{-1}(1-z)-F_{\rho}^{-1}(1-x)\Big)-\lam \frac{\alpha-z}{\alpha}\\
&=\ep\Big[(1+\ep) (F_{\rho}^{-1}(1-\alpha+\ep)-F_{\rho}^{-1}(1-\alpha-\ep^2) )- \frac{\lam}{\alpha}\Big]<0, \quad\mbox{as $\ep\to 0+$}.
\end{align*}
This indicates that $f(x)<0$ when $x$ is sufficiently close to $\alpha$. 
Geometrically speaking, the tangent line does not dominate $\varphi_{\lam}$ on $[0,\alpha]$. 
Because 
\begin{align}\label{mono2}
\frac{\partial}{\partial x}\big[\varphi_{\lam}(x)+F_{\rho}^{-1}(1-x)(z-x)-\varphi_{\lam}(z)\big]=(z-x)\frac{\partial}{\partial x}F_{\rho}^{-1}(1-x)>0, \quad z\in [0,\alpha],
\end{align}
we see that $f(x)$ is strictly increasing in $x$. 
Obviously, $f$ is also continuous on $(\alpha,1)$, so there exists a unique $z_2=z_2(\lam)\in (\alpha,1)$ such that $f(z_2)=0$, that is, 
\[\min_{z\in [0,\alpha]}\big[\varphi_{\lam}(z_2)+F_{\rho}^{-1}(1-z_2)(z-z_2)-\varphi_{\lam}(z)\big]=0.\]
Now fix this $z_2$. 
Because $\varphi_{\lam}$ is strictly concave on $[0,\alpha]$, the minimizer in above, denoted by $z_1=z_1(\lam)$, is unique and 
\begin{align} \label{z1=s(z2)}
\varphi_{\lam}(z_2)+F_{\rho}^{-1}(1-z_2)(z_1-z_2)-\varphi_{\lam}(z_1)=0.
\end{align} 
Because
\begin{align*} 
&\quad \varphi_{\lam}(z_2)+F_{\rho}^{-1}(1-z_2)(\alpha-z_2)-\varphi_{\lam}(\alpha)\\
&=-\int_{z_2}^1F_{\rho}^{-1}(1-y)\dy+F_{\rho}^{-1}(1-z_2)(\alpha-z_2)
+\int_{\alpha}^1F_{\rho}^{-1}(1-y)\dy\\
&=\int_{\alpha}^{z_2}F_{\rho}^{-1}(1-y)\dy+F_{\rho}^{-1}(1-z_2)(\alpha-z_2)\\ 
&=\int_{\alpha}^{z_2}\big(F_{\rho}^{-1}(1-y)-F_{\rho}^{-1}(1-z_2)\big)\dy>0,
\end{align*}
we see $z_1\neq \alpha$. Also, when $z$ is sufficiently close to 0, 
\begin{align*} 
\frac{\partial}{\partial z}\big[\varphi_{\lam}(z_2)+F_{\rho}^{-1}(1-z_2)(z-z_2)-\varphi_{\lam}(z)\big]&=F_{\rho}^{-1}(1-z_2)-F_{\rho}^{-1}(1-z)+\frac{\lam}{\alpha}<0, 
\end{align*} 
by recalling that $\rho$ is log-normal distributed, 
so $0$ is not the minimizer for the mapping 
\[z\mapsto\varphi_{\lam}(z_2)+F_{\rho}^{-1}(1-z_2)(z-z_2)-\varphi_{\lam}(z),\]
and $z_1\neq 0$. Thus, we conclude $z_1\in(0, \alpha)$. As it is the minimizer for the above mapping, it must satisfy the first order condition, so 
\[\varphi'_{\lam}(z_1)=F_{\rho}^{-1}(1-z_2)=\varphi'_{\lam}(z_2),\]
where the second equation is due to the definition of $\varphi'_{\lam}$ and $z_2 \in (\alpha,1)$. 
Meanwhile, the definition of $\varphi'_{\lam}$ and $z_1 \in (0,\alpha)$, we have 
\[\varphi'_{\lam}(z_1)=F_{\rho}^{-1}(1-z_1)-\frac{\lam}{\alpha},\]
so 
\[z_1=1-F_{\rho}\Big(F_{\rho}^{-1}(1-z_2)+\frac{\lam}{\alpha}\Big)=s(z_2).\]
Consequently, by \eqref{z1=s(z2)},
\begin{align*} 
&\quad\;\varphi_{\lam}(z_2)-\varphi_{\lam}(s(z_2))-\varphi'_{\lam}(z_2)(z_2-s(z_2))\\
&=\varphi_{\lam}(z_2)-\varphi_{\lam}(z_1)-F_{\rho}^{-1}(1-z_2)(z_2-z_1)=0.
\end{align*} 
Also
\begin{align*} 
\liminf_{\lam\to0}z_1(\lam)&=1-F_{\rho}\Big(\limsup_{\lam\to0}\Big( F_{\rho}^{-1}(1-z_2(\lam))+\frac{\lam}{\alpha}\Big)\Big)\\
&=1-F_{\rho}\Big( F_{\rho}^{-1}\Big(1-\liminf_{\lam\to0} z_2(\lam)\Big)\Big)\\
&= \liminf_{\lam\to0} z_2(\lam),
\end{align*}
and similarly, 
\begin{align*} 
\limsup_{\lam\to0}z_1(\lam) &= \limsup_{\lam\to0} z_2(\lam). 
\end{align*}
In view of $z_1(\lam)<\alpha<z_2(\lam)$, we get 
\[\lim_{\lam\to0}z_1(\lam)=\lim_{\lam\to0}z_2(\lam)=\alpha. \]

If $\lam\to\lamup$, then 
\[0\geq\varphi_{\lam}(z_2)\geq \max_{z\in [0,\alpha]}\big[\varphi_{\lam}(z_2)+F_{\rho}^{-1}(1-z_2)(z-z_2)\big]\geq \max_{z\in [0,\alpha]}\varphi_{\lam}(z)\to 0.\] 
This indicates $z_2(\lam)\to 1$.

The other claims are easy to verify. The proof is complete. 
\end{proof}

\begin{figure}[htbp]
\centering
\includegraphics[height=8cm,width=10cm]{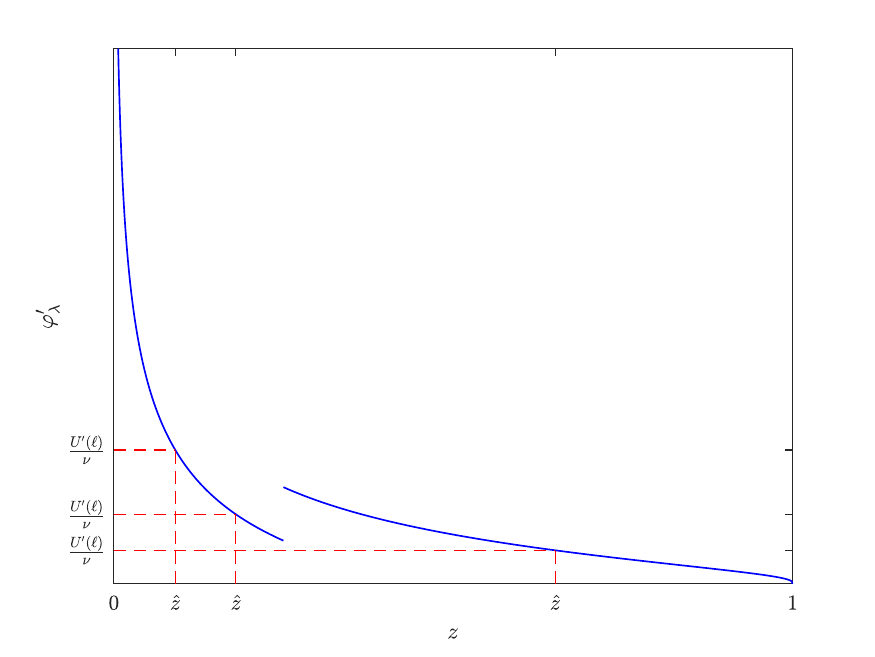}
\caption{The derivative of $\varphi_{\lambda}$.}
\label{varphi'}
\end{figure}

\begin{lemma}\label{hatzlemma}
Let
\[\hatz=\inf\Big\{z\in(0,1): \varphi'_{\lam}(z)<\frac{1}{\lamb}U'(\ell)\Big\}.\]
Then $$\varphi'_{\lam}(\hatz)=\frac{1}{\lamb}U'(\ell)$$ and 
\begin{align}\label{hatz}
\hatz=\begin{cases}
1-F_{\rho}\Big(\frac{\lam}{\alpha}+\frac{1}{\lamb} U'(\ell)\Big)\leq \alpha, &\quad
\text{if}\;F_{\rho}^{-1}(1-\alpha)\leq \frac{1}{\lamb} U'(\ell)+\frac{\lam}{\alpha};\medskip\\
1-F_{\rho}\big(\frac{1}{\lamb}U'(\ell)\big)>\alpha,&\quad\text{otherwise}. 
\end{cases}
\end{align}
\end{lemma}
\begin{proof}
Since
\[\varphi'_{\lam}(z)=F_{\rho}^{-1}(1-z)-\frac{\lam}{\alpha} I_{\{z\leq \alpha\}},\] 
we have $\varphi'_{\lam}(0+)=\infty$ and $\varphi'_{\lam}(1-)=0$ so that $0<\hatz<1$. 
Because $\varphi'_{\lam}$ firstly continuously decreases on $(0,\alpha)$, then jumps up and continuously decreases on $(\alpha,1)$ (see Figure \ref{varphi'}), 
we have $\varphi'_{\lam}(\hatz)=\frac{1}{\lamb}U'(\ell)$, namely 
\begin{align}\label{miniz}
F_{\rho}^{-1}(1-\hatz)-\frac{\lam}{\alpha} I_{\{\hatz\leq \alpha\}}=\frac{1}{\lamb}U'(\ell).
\end{align}
Clearly, the equation \eqref{miniz} has at most one solution in $(0, \alpha]$ and at most one in $(\alpha,1)$; and the smaller one is $\hatz$ if there are two solutions. 
There are two possible cases: 
\begin{itemize}
\item If 
\[ \min\{\varphi'_{\lam}(z): z\in (0,\alpha]\}=F_{\rho}^{-1}(1-\alpha)-\frac{\lam}{\alpha}\leq \frac{1}{\lamb} U'(\ell).\]
Then 
\[\hatz=1-F_{\rho}\Big(\frac{\lam}{\alpha}+\frac{1}{\lamb} U'(\ell)\Big),\] 
since it is unique (thus the smallest) number in $(0, \alpha]$ satisfying \eqref{miniz}. 
\item If 
\[ \min\{\varphi'_{\lam}(z): z\in (0,\alpha]\}=F_{\rho}^{-1}(1-\alpha)-\frac{\lam}{\alpha}> \frac{1}{\lamb} U'(\ell).\]
then \eqref{miniz} has no solution in $(0, \alpha]$. 
Hence $\hatz\in(\alpha,1)$, and thus, 
\[\hatz=1-F_{\rho}\big(\frac{1}{\lamb}U'(\ell)\big).\] 
\end{itemize} 
The proof is complete. 
\end{proof}

With the help of the above technique results, we are ready to present the optimal solution for the auxiliary problem \eqref{e8a}.

\begin{proposition}[Optimal solution for the auxiliary problem \eqref{e8a}]\label{main1} 
Let $\lamup$ be given in Lemma \ref{lamup} and let $(\lam,\lamb)\in (0, \lamup)\times(0,\infty)$. 
Let $z_1=z_1(\lam)\in (0,\alpha)$ and $z_2=z_2(\lam)\in(\alpha,1)$ be given in Lemma \ref{zonetwo}. 
Then the problem \eqref{e8a} admits a unique optimal solution $G^{*}_{\lam,\lamb}$ given as follows. 
\begin{itemize} 
\item
If $z_2 \geq 1-F_{\rho}\big(\frac{1}{\lamb}U'(\ell)\big)$, or equivalently, $z_1 \geq 1-F_{\rho}\Big(\frac{\lam}{\alpha}+\frac{1}{\lamb} U'(\ell)\Big)$, 
then 
\begin{align}\label{Gstara}
G^{*}_{\lam,\lamb}(z)=
\begin{cases}
0,&\quad 0<z\leq 1-F_{\rho}\Big(\frac{\lam}{\alpha}+\frac{1}{\lamb} U'(\ell)\Big),\medskip\\
(U')^{-1}(\lamb\delta'_{\lam,\lamb}(z))-\ell, &\quad 1-F_{\rho}\Big(\frac{\lam}{\alpha}+\frac{1}{\lamb} U'(\ell)\Big) <z<1. 
\end{cases}
\end{align}
where 
\begin{equation}\label{deltaprime1}
\delta_{\lam,\lamb}'(z)=
\begin{cases}
F_{\rho}^{-1}(1-z)-\frac{\lam}{\alpha},& 1-F_{\rho}\Big(\frac{\lam}{\alpha}+\frac{1}{\lamb} U'(\ell)\Big)<z\leq z_1,\medskip\\
F_{\rho}^{-1}(1-z_2),&z_1< z\leq z_2,\medskip\\
F_{\rho}^{-1}(1-z),&z_2<z<1.
\end{cases} 
\end{equation}
\item 
If $z_2 < 1-F_{\rho}\big(\frac{1}{\lamb}U'(\ell)\big)$, or equivalently,
$z_1 < 1-F_{\rho}\Big(\frac{\lam}{\alpha}+\frac{1}{\lamb} U'(\ell)\Big)$, 
then 
\begin{align}\label{Gstarc}
G^{*}_{\lam,\lamb}(z)&=
\begin{cases}
0, &\quad 0<z\leq 1-F_{\rho}\big(\frac{1}{\lamb}U'(\ell)\big),\medskip\\
(U')^{-1}(\lamb F_{\rho}^{-1}(1-z))-\ell, &\quad 1-F_{\rho}\big(\frac{1}{\lamb}U'(\ell)\big)<z<1,
\end{cases}\nn\medskip\\
&=\big((U')^{-1}(\lamb F_{\rho}^{-1}(1-z))-\ell\big)^+.
\end{align}
\end{itemize}
\end{proposition}
\begin{proof}
Let $\hatz$ be defined in Lemma \ref{hatzlemma}. 
If $z\in (0, \hatz]$, then $U'(x+\ell)\leq U'(\ell)\leq \lamb \varphi'_{\lam}(z)$ for $x\geq 0$, so the mapping 
\[x\mapsto U(x+\ell)-\lamb x \varphi'_{\lam}(z) \]
is non-increasing on $[0,\infty)$. Thus 
\begin{align*} 
\int_0^{\hatz}\big(U(G(z)+\ell)-\lamb G(z)\varphi'_{\lam}(z)\big)\dz 
&\leq \int_0^{\hatz}\big(U(0+\ell)-\lamb *0*\varphi'_{\lam}(z)\big)\dz=U(\ell)\hatz,
\end{align*} 
and 
\begin{align*} 
\int_0^{1}\big(U(G(z)+\ell)-\lamb G(z)\varphi'_{\lam}(z)\big)\dz 
&\leq U(\ell)\hatz+\int_{\hatz}^{1}\big(U(G(z)+\ell)-\lamb G(z)\varphi'_{\lam}(z)\big)\dz.
\end{align*} 
We only need to solve 
\begin{align} \label{e8}
\underset{G\in\setg}{\sup}\int_{\hatz}^{1}\big(U(G(z)+\ell)-\lamb G(z)\varphi'_{\lam}(z)\big)\dz.
\end{align} 
This problem has been indeed solved by Xia and Zhou \cite{XZ16} and Xu \cite{X16} by different approaches. 

Because $z_2\in (\alpha,1)$, by monotonicity, 
\[0<\varphi'_{\lam}(z_2)<\varphi'_{\lam}(\alpha)=F_{\rho}^{-1}(1-\alpha).\]
Since $\varphi'_{\lam}(z_1)=\varphi'_{\lam}(z_2)$, we have 
\[0<\varphi'_{\lam}(z_1)=F_{\rho}^{-1}(1-z_1)-\frac{\lam}{\alpha}<F_{\rho}^{-1}(1-\alpha).\]
There are three possible cases. 
\begin{itemize} 
\item If $F_{\rho}^{-1}(1-z_1) \leq\frac{1}{\lamb} U'(\ell)+\frac{\lam}{\alpha}$,
then since $z_1\leq \alpha$, we have $F_{\rho}^{-1}(1-\alpha) \leq\frac{1}{\lamb} U'(\ell)+\frac{\lam}{\alpha}$. Thus by \eqref{hatz}, 
\[\hatz=1-F_{\rho}\Big(\frac{\lam}{\alpha}+\frac{1}{\lamb} U'(\ell)\Big).
\]
The condition also shows $\hatz\leq z_1$. 
Let $\delta_{\lam,\lamb}$ denote the concave envelope of $\varphi_{\lam}$ on $[0,1]$. 
Then by Lemma \ref{zonetwo}, $\delta_{\lam,\lamb}=\varphi_{\lam}$ on $[0,z_1]\cup [z_2,1]$, and $\delta_{\lam,\lamb}$ is affine on $[z_1,z_2]$, so $\delta'_{\lam,\lamb}$ satisfies \eqref{deltaprime1}. Since $\delta_{\lam,\lamb}(\hatz)=\varphi_{\lam}(\hatz)$, $\delta_{\lam,\lamb}(1)=\varphi_{\lam}(1)$, it yields 
\begin{align*} 
&\quad\;\int_{\hatz}^{1}G(z)(\varphi'_{\lam}(z)-\delta'_{\lam,\lamb}(z))\dz\\
&=G(z)(\varphi_{\lam}(z)-\delta_{\lam,\lamb}(z))\Big|_{z=\hatz}^{1}-\int_{\hatz}^{1}(\varphi_{\lam}(z)-\delta_{\lam,\lamb}(z))d G(z)\\
&\geq 0,
\end{align*} 
so 
\begin{align*} 
&\quad\;\int_{\hatz}^{1}\big(U(G(z)+\ell)-\lamb G(z)\varphi'_{\lam}(z)\big)\dz\\
&\leq \int_{\hatz}^{1}\big(U(G(z)+\ell)-\lamb G(z)\delta'_{\lam,\lamb}(z)\dz\\
&\leq \int_{\hatz}^{1}\big(U(G^{*}_{\lam,\lamb}(z)+\ell)-\lamb G^{*}_{\lam,\lamb}(z)\delta'_{\lam,\lamb}(z)\dz,
\end{align*} 
where the last inequality is due to point wise optimization, and 
\begin{align*} 
G^{*}_{\lam,\lamb}(z)=\max\big\{(U')^{-1}(\lamb\delta'_{\lam,\lamb}(z))-\ell,\;0\big\}.
\end{align*} 
But $\delta'_{\lam,\lamb}(z)$ is decreasing, so
\[\delta'_{\lam,\lamb}(z)\leq \delta'_{\lam,\lamb}(\hatz)=\varphi'_{\lam}(\hatz)=\frac{1}{\lamb} U'(\ell),\quad z\geq \hatz,\]
and
\begin{align*} 
G^{*}_{\lam,\lamb}(z)=(U')^{-1}(\lamb\delta'_{\lam,\lamb}(z))-\ell,\quad z\geq \hatz.
\end{align*} 
Therefore, \eqref{Gstara} is the optimal solution to the problem \eqref{e8a}. 
\item If $F_{\rho}^{-1}(1-\alpha)< \frac{1}{\lamb} U'(\ell)+\frac{\lam}{\alpha}<F_{\rho}^{-1}(1-z_1)$, then by \eqref{hatz}, 
\[\hatz=1-F_{\rho}\Big(\frac{\lam}{\alpha}+\frac{1}{\lamb} U'(\ell)\Big).
\]
The condition also gives $z_1<\hatz<\alpha$. 
Consider the class of tangent lines of $\varphi_{\lam}$ with the points of tangency $(z,\varphi_{\lam}(z))$ for $z\in (\alpha, z_2)$. Each of these lines meets the curve $\varphi_{\lam}$ at some point with abscissa in $(z_1,\alpha)$, and vise versa. 
In particular, one of the meeting points has abscissa $\hatz$. Let $(z_3,\varphi_{\lam}(z_3))$ be the corresponding point of tangency. Let $\delta_{\lam,\lamb}$ denote the concave envelope of $\varphi_{\lam}$ on $[\hatz,1]$. Then $\delta_{\lam,\lamb}$ is affine on $[\hatz,z_3]$ and coincides with $\varphi_{\lam}$ on $[z_3,1]$. Hence, 
\begin{equation*}
\delta_{\lam,\lamb}'(z)=
\begin{cases}
F_{\rho}^{-1}(1-z_3),&\quad \hatz< z\leq z_3,\medskip\\
F_{\rho}^{-1}(1-z),&\quad z_3<z<1.\medskip
\end{cases} 
\end{equation*} 
By the same preceding argument, 
\begin{align*} 
&\quad\;\int_{\hatz}^{1}\big(U(G(z)+\ell)-\lamb G(z)\varphi'_{\lam}(z)\big)\dz\\
&\leq \int_{\hatz}^{1}\big(U(G(z)+\ell)-\lamb G(z)\delta'_{\lam,\lamb}(z)\big)\dz\\
&\leq \int_{\hatz}^{1}\big(U(G^{*}_{\lam,\lamb}(z)+\ell)-\lamb G^{*}_{\lam,\lamb}(z)\delta'_{\lam,\lamb}(z)\big)\dz.
\end{align*} 
where the second inequality is due to the point wise optimization, 
\begin{align*} 
G^{*}_{\lam,\lamb}(z)=\max\big\{(U')^{-1}(\lamb\delta'_{\lam,\lamb}(z))-\ell,\;0\big\}.
\end{align*} 

Since $0<z_1<\hatz<\alpha$ and $\varphi'_{\lam}$ is strictly decreasing on $(0,\alpha)$, 
\[\varphi'_{\lam}(z_1)>\varphi'_{\lam}(\hatz)=\frac{1}{\lamb} U'(\ell).\]
Also, since $\varphi'_{\lam}$ is strictly decreasing on $(\alpha,1)$, $\alpha<z_3<z_2<1$, and $\varphi'_{\lam}(z_2)=\varphi'_{\lam}(z_1)$, we have
\[F_{\rho}^{-1}(1-\alpha)> F_{\rho}^{-1}(1-z_3)=\varphi'_{\lam}(z_3)>\varphi'_{\lam}(z_2)=\varphi'_{\lam}(z_1)>\frac{1}{\lamb} U'(\ell).\] 
Hence, $$z_4=1-F_{\rho}\big(\frac{1}{\lamb}U'(\ell)\big)\in (\alpha,1).$$ Since $z_3$, $z_4\in (\alpha,1)$ and 
\[\varphi'_{\lam}(z_4)=\frac{1}{\lamb} U'(\ell)<\varphi'_{\lam}(z_3),\]
it follows from the monotonicity of $\varphi'_{\lam}$ on $ (\alpha,1)$ that $z_4>z_3$.
Hence, by monotonicity, 
\[(U')^{-1}(\lamb\delta'_{\lam,\lamb}(z))\leq (U')^{-1}(\lamb\delta'_{\lam,\lamb}(z_4))
=(U')^{-1}(\lamb F_{\rho}^{-1}(1-z_4))=\ell,\]
for any $z\in [\hatz,z_4]$, 
and 
\[(U')^{-1}(\lamb\delta'_{\lam,\lamb}(z))
\geq (U')^{-1}(\lamb\delta'_{\lam,\lamb}(z_4))=\ell,\]
for any $z\in (z_4,1)$. 
Therefore,
\begin{align*}
G^{*}_{\lam,\lamb}(z)&=\max\big\{(U')^{-1}(\lamb\delta'_{\lam,\lamb}(z))-\ell,\;0\big\}\\
&=
\begin{cases}
0, &\quad \hatz<z\leq z_4,\medskip\\
(U')^{-1}(\lamb F_{\rho}^{-1}(1-z))-\ell, &\quad z_4<z<1. 
\end{cases}
\end{align*}
This shows shows that \eqref{Gstarc} is the optimal solution to the problem \eqref{e8a}.

\item If $F_{\rho}^{-1}(1-\alpha)\geq \frac{1}{\lamb} U'(\ell)+\frac{\lam}{\alpha}$, then by \eqref{hatz}, 
\[\hatz=
\begin{cases}
1-F_{\rho}\Big(\frac{\lam}{\alpha}+\frac{1}{\lamb} U'(\ell)\Big)=\alpha, &\quad\mbox{if\quad}F_{\rho}^{-1}(1-\alpha)=\frac{1}{\lamb} U'(\ell)+\frac{\lam}{\alpha},\medskip\\
1-F_{\rho}\Big(\frac{1}{\lamb}U'(\ell)\Big)>\alpha, &\quad\mbox{if\quad} F_{\rho}^{-1}(1-\alpha)>\frac{1}{\lamb} U'(\ell)+\frac{\lam}{\alpha}.
\end{cases}
\]

For each $z\in (\hatz, 1)$, we have $z>\alpha$ and thus $\varphi'_{\lam}(z)=F_{\rho}^{-1}(1-z)$,
so the map
\[x\mapsto U(x+\ell)-\lamb x \varphi'_{\lam}(z)=U(x+\ell)-\lamb x F_{\rho}^{-1}(1-z)\]
is maximized at $\max\{(U')^{-1}(\lamb F_{\rho}^{-1}(1-z))-\ell, 0\}$ on $[0,\infty)$. By monotonicity, we see 
$(U')^{-1}(\lamb F_{\rho}^{-1}(1-z))-\ell > 0$ if and only if 
$z> 1-F_{\rho}\big(\frac{1}{\lamb}U'(\ell)\big)$. 
Since $1-F_{\rho}\big(\frac{1}{\lamb}U'(\ell)\big)\geq \hatz$, we conclude that \eqref{Gstarc} 
is the optimal solution to the problem \eqref{e8a}. 
\end{itemize} 
The proof is complete. 
\end{proof} 

\begin{remark}
Notice $\hatz=0$ if $\ell=0$. Therefore, the problem is essentially different when $\ell>0$. 
\end{remark}

\begin{remark}
For the upper bounded \tvar constraint model, corresponding to \eqref{e8a}, we have 
\[\varphi'_{\lam}(z)=F_{\rho}^{-1}(1-z)+\frac{\lam}{\alpha} I_{\{z\leq \alpha\}}\] 
which is strictly decreasing. The problem \eqref{e8a} is well-posed and the optimal solution is given by 
\[G^{*}_{\lam,\lamb}(z)=\max\big\{(U')^{-1}(\lamb\varphi'_{\lam}(z))-\ell,\; 0\big\}.\]
Indeed, the map
\[x\mapsto U(x+\ell)-\lamb x \varphi'_{\lam}(z)\]
is maximized at $G^{*}_{\lam,\lamb}(z)$ on $[0,\infty).$ The model is relatively easier to study. 
\end{remark}
 \begin{theorem}
[Optimal solution for the problem \eqref{e4} with an effective \ltvar constraint] 
\label{main2}
Let $\lamup$ be given in Lemma \ref{lamup} and let $(\lam,\lamb)\in (0, \lamup)\times(0,\infty)$. 
Let $G^{*}_{\lam,\lamb}$ denote the unique optimal solution to the problem \eqref{e8a} given in Theorem \ref{main1}.
If $R(\underz)<\barz<C(\underz)$, then there exists a pair of Lagrange multiplier $(\lam^*,\lamb^*)\in (0, \lamup) \times(0,\infty)$ such that 
\[\int_0^1 G^{*}_{\lam^{*},\lamb^{*}}(z)F_{\rho}^{-1}(1-z)\dz=\barz,\quad 
\frac{1}{\alpha}\int_{0}^{\alpha}G^{*}_{\lam^{*},\lamb^{*}}(z)\dz=\underz.\]
Furthermore, $G^{*}_{\lam^{*},\lamb^{*}}$ is the optimal solution to the problem \eqref{e4}.
\end{theorem} 
Its proof will be given after the proof of Lemma \ref{lagrangeexist}. 

To prove the existence of the Lagrange multipliers, 
 we define 
\begin{align*}
f(\lam,\lamb)=\int_0^1 G^{*}_{\lam,\lamb}(z)F_{\rho}^{-1}(1-z)\dz,\quad 
g(\lam,\lamb)=\frac{1}{\alpha}\int_{0}^{\alpha}G^{*}_{\lam,\lamb}(z)\dz.
\end{align*}
We present their properties in the following lemmas. 
\begin{lemma} \label{propertyg}
Let $\lamup$ be given in Lemma \ref{lamup}. 
For each $\lam\in (0, \lamup)$, the function $g(\lam,\cdot)$ continuous and strictly decreasing on $(0,\infty)$ with 
\begin{align}\label{glimit1}
\lim\limits_{\lamb\to0^{+}}g(\lam,\lamb)=+\infty, \quad \lim\limits_{\lamb\to+\infty}g(\lam,\lamb)=0.
\end{align}
For each $\lamb\in (0,\infty)$, the function $g(\cdot,\lamb)$ is continuous and increasing on $ (0, \lamup)$ with
\begin{align}\label{glimit2}
\lim_{\lam\to0^{+}}g(\lam,\lamb)&=\frac{1}{\alpha}\int_{0}^{\alpha} \big((U')^{-1}(\lamb F_{\rho}^{-1}(1-z))-\ell\big)^+\dz,\quad 
\lim_{\lam\to \lamup}g(\lam,\lamb)=+\infty.
\end{align}
\end{lemma}

\begin{proof}
For each fixed $\lam\in (0, \lamup)$, 
let $z_1=z_1(\lam)\in (0,\alpha)$ and $z_2=z_2(\lam)\in(\alpha,1)$ be given in Lemma \ref{zonetwo}. Then $z_1$ and $z_2$ are fixed. So for each fixed $z\in(0,1)$, by Theorem \ref{main1}, the function $G^{*}_{\lam,\cdot}(z)$ is continuously and strictly decreasing on 
$ (0,\frac{U'(\ell)}{F_{\rho}^{-1}(1-z_2)} ]$ and on $ (\frac{U'(\ell)}{F_{\rho}^{-1}(1-z_2)},\infty)$ respectively. 
We now show it is continuous at $\lamb_0=\frac{U'(\ell)}{F_{\rho}^{-1}(1-z_2)}$ as well. At this point, we have $z_2 = 1-F_{\rho}\Big(\frac{1}{\lamb_0} U'(\ell)\Big)$,
which is equivalent to $z_1 = 1-F_{\rho}\Big(\frac{\lam}{\alpha}+\frac{1}{\lamb_0} U'(\ell)\Big)$, 
so \eqref{Gstara} reads 
\begin{align*}
G^{*}_{\lam,\lamb_0}(z)=
\begin{cases}
0, &\quad 0<z\leq z_1,\medskip\\
(U')^{-1}(\lamb_0\delta'_{\lam,\lamb_0}(z))-\ell, &\quad z_1<z<1,
\end{cases}
\end{align*}
where 
\begin{equation*}
\delta_{\lam,\lamb_0}'(z)=
\begin{cases}
F_{\rho}^{-1}(1-z_2)=\frac{1}{\lamb_0} U'(\ell),&z_1< z\leq z_2=1-F_{\rho}\Big(\frac{1}{\lamb_0} U'(\ell)\Big),\medskip\\
F_{\rho}^{-1}(1-z),&z_2=1-F_{\rho}\Big(\frac{1}{\lamb_0} U'(\ell)\Big)<z<1.
\end{cases} 
\end{equation*}
After simplification, 
\begin{align*}
G^{*}_{\lam,\lamb_0}(z)=
\begin{cases}
0, &\quad 0<z\leq 1-F_{\rho}\Big(\frac{1}{\lamb_0} U'(\ell)\Big),\medskip\\
(U')^{-1}(\lamb_0 F_{\rho}^{-1}(1-z))-\ell, &\quad 1-F_{\rho}\Big(\frac{1}{\lamb_0} U'(\ell)\Big)<z<1,
\end{cases}
\end{align*}
Comparing with \eqref{Gstarc}, we conclude that the function $G^{*}_{\lam,\cdot}(z)$ is continuous at $\lamb_0$. Hence, by the monotone convergence theorem, the function $g(\lam,\cdot)$ is continuous and strictly decreasing on $(0,\infty)$.

The two limits in \eqref{glimit1} are the consequence of the monotone convergence theorem and the Inada conditions on $U$.

By Theorem \ref{main1}, when $z_2 \geq 1-F_{\rho}\big(\frac{1}{\lamb}U'(\ell)\big)$, 
\begin{align*}
g(\lam,\lamb)&=\frac{1}{\alpha}\Bigg[\int_{z_5}^{z_1} (U')^{-1}\Big(\lamb\Big(F_{\rho}^{-1}(1-z)-\frac{\lam}{\alpha}\Big)\Big)\dz\\
&\quad\;+\int_{z_1}^{\alpha} (U')^{-1}\Big(\lamb\Big(F_{\rho}^{-1}(1-z_2)\Big)\Big)\dz\Bigg],
\end{align*}
where
\[z_5=1-F_{\rho}\Big(\frac{\lam}{\alpha}+\frac{1}{\lamb} U'(\ell)\Big).\]
Because 
\[F_{\rho}^{-1}(1-z_1)-\frac{\lam}{\alpha}=\varphi'_{\lam}(z_1)=\varphi'_{\lam}(z_2)=F_{\rho}^{-1}(1-z_2),\] 
by the chain rule, it follows
\begin{align*}
\frac{\partial}{\partial\lam}g(\lam,\lamb)&=\frac{1}{\alpha}\Bigg[\int_{z_5}^{z_1} \frac{\partial}{\partial\lam}(U')^{-1}\Big(\lamb\Big(F_{\rho}^{-1}(1-z)-\frac{\lam}{\alpha}\Big)\Big)\dz\\
&\quad\;+ (U')^{-1}\Big(\lamb\Big(F_{\rho}^{-1}(1-z_1)-\frac{\lam}{\alpha}\Big)\Big) \frac{\partial z_1}{\partial\lam}\\
&\quad\;-(U')^{-1}\Big(\lamb\Big(F_{\rho}^{-1}(1-z_5)-\frac{\lam}{\alpha}\Big)\Big) \frac{\partial z_5}{\partial\lam}\\
&\quad\;-(U')^{-1}\Big(\lamb\Big(F_{\rho}^{-1}(1-z_2)\Big)\Big) \frac{\partial z_1}{\partial\lam}\Bigg]\\
&=\frac{1}{\alpha}\Bigg[\int_{z_5}^{z_1} \frac{\partial}{\partial\lam}(U')^{-1}\Big(\lamb\Big(F_{\rho}^{-1}(1-z)-\frac{\lam}{\alpha}\Big)\Big)\dz\\
&\quad\;-(U')^{-1}\Big(\lamb\Big(F_{\rho}^{-1}(1-z_5)-\frac{\lam}{\alpha}\Big)\Big) \frac{\partial z_5}{\partial\lam}\Bigg].
\end{align*}
Clearly \[\frac{\partial}{\partial\lam}(U')^{-1}\Big(\lamb\Big(F_{\rho}^{-1}(1-z)-\frac{\lam}{\alpha}\Big)\Big)>0,\quad \frac{\partial z_5}{\partial\lam}<0,\] 
so $g(\cdot,\lamb)$ is strictly increasing. 

When $z_2 < 1-F_{\rho}\big(\frac{1}{\lamb}U'(\ell)\big)$, clearly $g(\cdot,\lamb)$ is a constant. Therefore, to show $g(\cdot,\lamb)$ is increasing on $(0,\lamup)$, 
it suffices to show that $g(\cdot,\lamb)$ is continuous at the point $\lam_0$ such that $z_2(\lam_0)=1-F_{\rho}\big(\frac{1}{\lamb}U'(\ell)\big)$. 
In this case, 
\[z_1(\lam_0)=1-F_{\rho}\Big(\frac{\lam_0}{\alpha}+\frac{1}{\lamb} U'(\ell)\Big).\]
Taking this into \eqref{Gstara} yields 
\begin{align*}
G^{*}_{\lam_0,\lamb}(z)=
\begin{cases}
0,&\quad 0<z\leq z_1,\medskip\\
(U')^{-1}(\lam_0\delta'_{\lam_0,\lamb}(z))-\ell, &\quad z_1<z<1. 
\end{cases}
\end{align*}
where 
\begin{equation*}
\delta_{\lam_0,\lamb}'(z)=
\begin{cases}
F_{\rho}^{-1}(1-z_2),&z_1< z\leq z_2,\medskip\\
F_{\rho}^{-1}(1-z),&z_2<z<1.
\end{cases} 
\end{equation*}
When $z_1< z\leq z_2$, it follows
\begin{align*}
G^{*}_{\lam_0,\lamb}(z)= (U')^{-1}(\lam_0F_{\rho}^{-1}(1-z_2))-\ell=0.
\end{align*}
Hence, 
\begin{align*}
G^{*}_{\lam_0,\lamb}(z)=
\begin{cases}
0, &\quad 0<z\leq z_2=1-F_{\rho}\big(\frac{1}{\lamb}U'(\ell)\big),\medskip\\
(U')^{-1}(\lamb F_{\rho}^{-1}(1-z))-\ell, &\quad z_2=1-F_{\rho}\big(\frac{1}{\lamb}U'(\ell)\big)<z<1.
\end{cases}
\end{align*}
Comparing with \eqref{Gstarc}, we conclude the function $G^{*}_{\cdot,\lamb}(z)$ is continuous at $\lam_0$. As a consequence, $g(\cdot,\lamb)$ is continuous (and thus increasing) on $(0,\lamup)$. 

It is left to prove the two limits in \eqref{glimit2}. 
In view of Lemma \ref{zonetwo}, 
\[\lim_{\lam\to0}z_1(\lam)=\lim_{\lam\to0}z_2(\lam)=\alpha.\]
Using \eqref{Gstara} if $\alpha > 1-F_{\rho}\big(\frac{1}{\lamb}U'(\ell)\big)$,
and using \eqref{Gstarc} if $\alpha <1-F_{\rho}\big(\frac{1}{\lamb}U'(\ell)\big)$, 
we get 
\begin{align*}
\lim_{\lam\to0^{+}}G^{*}_{\lam,\lamb}(z)=
\begin{cases}
0,&\quad 0<z\leq 1-F_{\rho}\big(\frac{1}{\lamb}U'(\ell)\big),\medskip\\
(U')^{-1}(\lamb F_{\rho}^{-1}(1-z))-\ell, &\quad 1-F_{\rho}\big(\frac{1}{\lamb}U'(\ell)\big) <z<1, 
\end{cases}
\end{align*} 
so 
\begin{align*}
\lim_{\lam\to0^{+}}g(\lam,\lamb) &=\frac{1}{\alpha}\int_{1-F_{\rho}\big(\frac{1}{\lamb}U'(\ell)\big) }^{\alpha} \big((U')^{-1}(\lamb F_{\rho}^{-1}(1-z))-\ell\big)\dz\\
&=\frac{1}{\alpha}\int_{0}^{\alpha} \big((U')^{-1}(\lamb F_{\rho}^{-1}(1-z))-\ell\big)^+\dz.
\end{align*}

Similarly, by Lemma \ref{zonetwo}, 
\[ \lim_{\lam\to\lamup}z_1(\lam)=1-F_{\rho}\big(\tfrac{\lamup}{\alpha}\big),\quad \lim_{\lam\to\lamup}z_2(\lam)=1>1-F_{\rho}\big(\frac{1}{\lamb}U'(\ell)\big).\] 
So we can use \eqref{Gstara} to get 
\begin{align*} 
\lim_{\lam\to\lamup}G^{*}_{\lam,\lamb}(z)=
\begin{cases}
0,&\quad 0<z\leq 1-F_{\rho}\Big(\frac{\lam}{\alpha}+\frac{1}{\lamb} U'(\ell)\Big),\medskip\\
(U')^{-1}\big(\lamb \big(F_{\rho}^{-1}(1-z)-\frac{\lam}{\alpha}\big)\big)-\ell, &\quad 1-F_{\rho}\Big(\frac{\lam}{\alpha}+\frac{1}{\lamb} U'(\ell)\Big) <z\leq 1-F_{\rho}\big(\tfrac{\lamup}{\alpha}\big)\medskip\\
+\infty, &\quad 1-F_{\rho}\big(\tfrac{\lamup}{\alpha}\big)< z<1.
\end{cases}
\end{align*}
so 
$\lim_{\lam\to \lamup}g(\lam,\lamb)=+\infty.$ 
The proof is complete. 
\end{proof}

\begin{lemma} \label{propertyf}
Let $\lamup$ be given in Lemma \ref{lamup}. 
For each $\lam\in (0, \lamup)$, the function $f(\lam,\cdot)$ is continuous and strictly decreasing on $(0,\infty)$ with 
\begin{align}\label{flimit1}
\lim\limits_{\lamb\to0^{+}}f(\lam,\lamb)=\infty, \quad \lim\limits_{\lamb\to+\infty}f(\lam,\lamb)=0.
\end{align}
For each $\lamb\in (0,\infty)$, the function $f(\cdot,\lamb)$ is continuous and increasing on $ (0, \lamup)$. Moreover, 
\begin{align}\label{flimit2}
\lim_{\lam\to0^{+}}f(\lam,\lamb)
&=\int_{0}^1 \big((U')^{-1}(\lamb F_{\rho}^{-1}(1-z))-\ell\big)^+F_{\rho}^{-1}(1-z)\dz,
\quad \lim_{\lam\to \lamup}f(\lam,\lamb)=+\infty.
\end{align}
\end{lemma}

\begin{proof}
The proof is similar to that of Lemma \ref{propertyg}, so we omit the details. 
\end{proof}

Recall that we need to deal with the case: $R(\underz)<\barz<C(\underz)$. In this case $\underz>0$.

Thanks to Lemma \ref{propertyg}, for each $\lam\in (0, \lamup)$, 
there exists a unique $\lamb\in(0,\infty)$, denoted by $h(\lam,\underz)$, such that $g(\lam,\lamb)=\underz $.
Also, the function $h(\cdot,\underz )$ is continuous and increasing on $(0,\lamup)$ with 
\begin{align}\label{hlimit1}
\lim\limits_{\lam\to0^{+}}h(\lam,\underz)&=\underlam(\underz),\quad \lim\limits_{\lam\to\lamup }h(\lam,\underz)=+\infty,
\end{align} 
where $\underlam(\underz)>0$ is determined by \eqref{functionC2}. 

\begin{lemma} \label{lagrangeexist}
We have 
\begin{align}\label{fhlimit1}
\lim_{\lam\to 0+} f(\lam,h(\lam,\underz ))&=C(\underz),
\end{align}
and
\begin{align}\label{fhlimit2}
\lim_{\lam\to \lamup} f(\lam,h(\lam,\underz ))&=R(\underz).
\end{align}
\end{lemma}
\begin{proof}
In view of Lemma \ref{propertyf}, $f(\cdot,h(\cdot,\underz ))$ is continuous on $(0,\lamup)$. By \eqref{hlimit1}, 
\begin{align*} 
\lim_{\lam\to 0+} f(\lam,h(\lam,\underz ))&=\lim_{\lam\to 0+} f(\lam, \underlam(\underz) )\\
&=\int_{0}^1 \big((U')^{-1}(\underlam(\underz) F_{\rho}^{-1}(1-z))-\ell\big)^+F_{\rho}^{-1}(1-z)\dz
= C(\underz).
\end{align*}
where the last inequality is due to the definition \eqref{functionC}, completing the proof of \eqref{fhlimit1}.

On the other hand, 
by \eqref{Gstara}, we have 
\begin{align*} 
G^{*}_{\lam,h(\lam,\underz)}(z)=
\begin{cases}
0,&\quad 0<z\leq 1-F_{\rho}\Big(\frac{\lam}{\alpha}+\frac{1}{h(\lam,\underz)} U'(\ell)\Big),\medskip\\
(U')^{-1}(\lamb\delta'_{\lam,h(\lam,\underz)}(z))-\ell, &\quad 1-F_{\rho}\Big(\frac{\lam}{\alpha}+\frac{1}{h(\lam,\underz)} U'(\ell)\Big) <z<1. 
\end{cases}
\end{align*}
where 
\begin{equation*}
\delta_{\lam,h(\lam,\underz)}'(z)=
\begin{cases}
F_{\rho}^{-1}(1-z)-\frac{\lam}{\alpha},& 1-F_{\rho}\Big(\frac{\lam}{\alpha}+\frac{1}{h(\lam,\underz)} U'(\ell)\Big)<z\leq z_1(\lam),\medskip\\
F_{\rho}^{-1}(1-z_2(\lam)),&z_1(\lam)< z\leq z_2(\lam),\medskip\\
F_{\rho}^{-1}(1-z),&z_2(\lam)<z<1.
\end{cases} 
\end{equation*}
Write 
\[z(\lam)=1-F_{\rho}\Big(\frac{\lam}{\alpha}+\frac{1}{h(\lam,\underz)} U'(\ell)\Big). \]
By Lemma \ref{zonetwo}, 
\[ \lim_{\lam\to\lamup}z_1(\lam)=\lim_{\lam\to\lamup}z(\lam)=1-F_{\rho}\big(\tfrac{\lamup}{\alpha}\big)<\alpha<1=\lim_{\lam\to\lamup}z_2(\lam).\] 
By the monotonicity of quantiles, we have 
\begin{align*} 
\alpha \underz &=\alpha g(\lam,h(\lam,\underz))=\int_{0}^{\alpha}G^{*}_{\lam,h(\lam,\underz)}(z)\dz=\int_{z(\lam)}^{\alpha}G^{*}_{\lam,h(\lam,\underz)}(z)\dz \\
&\leq G^{*}_{\lam,h(\lam,\underz)}(z_2(\lam))(\alpha-z(\lam)) =\Big[(U')^{-1}(h(\lam,\underz) F_{\rho}^{-1}(1-z_2(\lam)))-\ell\Big](\alpha-z(\lam)),
\end{align*} 
and
\begin{align*} 
\alpha\underz &=\int_{0}^{\alpha}G^{*}_{\lam,h(\lam,\underz)}(z)\dz\geq \int_{z_1(\lam)}^{\alpha}G^{*}_{\lam,h(\lam,\underz)}(z)\dz \\
&\geq G^{*}_{\lam,h(\lam,\underz)}(z_1(\lam)+)(\alpha-z_1(\lam)) =\Big[(U')^{-1}(h(\lam,\underz) F_{\rho}^{-1}(1-z_2(\lam)))-\ell\Big](\alpha-z_1(\lam)).
\end{align*} 
The above two inequalities imply
\begin{align*} 
\lim_{\lam\to\lamup}\Big[(U')^{-1}(h(\lam,\underz) F_{\rho}^{-1}(1-z_2(\lam)))-\ell\Big] 
=\frac{\alpha\underz}{\alpha-1+F_{\rho}\big(\tfrac{\lamup}{\alpha}\big)}.
\end{align*} 

By the monotonicity of quantiles again, we have 
\begin{align*}
f(\lam,h(\lam,\underz))&=\int_0^1 G^{*}_{\lam,h(\lam,\underz)}(z)F_{\rho}^{-1}(1-z)\dz
=\int_{z(\lam)}^1 G^{*}_{\lam,h(\lam,\underz)}(z)F_{\rho}^{-1}(1-z)\dz \\
&\leq \Big[(U')^{-1}(h(\lam,\underz) F_{\rho}^{-1}(1-z_2(\lam)))-\ell\Big]\int_{z(\lam)}^{z_2(\lam)}F_{\rho}^{-1}(1-z)\dz\\ 
&\quad\;+\int_{z_2(\lam)}^1\big((U')^{-1}(h(\lam,\underz) F_{\rho}^{-1}(1-z)))-\ell\big)F_{\rho}^{-1}(1-z)\dz.
\end{align*}
and
\begin{align*}
f(\lam,h(\lam,\underz)) 
&\geq \Big[(U')^{-1}(h(\lam,\underz) F_{\rho}^{-1}(1-z_2(\lam)))-\ell\Big]\int_{z_1(\lam)}^{z_2(\lam)}F_{\rho}^{-1}(1-z)\dz\\ 
&\quad\;+\int_{z_2(\lam)}^1\big((U')^{-1}(h(\lam,\underz) F_{\rho}^{-1}(1-z)))-\ell\big)F_{\rho}^{-1}(1-z)\dz.
\end{align*}
Suppose 
\begin{align} \label{zerolim}
\lim_{\lam\to\lamup}\int_{z_2(\lam)}^1\big((U')^{-1}(h(\lam,\underz) F_{\rho}^{-1}(1-z)))-\ell\big)F_{\rho}^{-1}(1-z)\dz=0,
\end{align}
then 
\begin{align*}
\lim_{\lam\to\lamup}f(\lam,h(\lam,\underz))
&= \frac{\alpha\underz}{\alpha-1+F_{\rho}\big(\tfrac{\lamup}{\alpha}\big)}
\int_{1-F_{\rho}\big(\tfrac{\lamup}{\alpha}\big)}^{1}F_{\rho}^{-1}(1-z)\dz
=\lamup \underz=R(\underz), 
\end{align*}
thanks to Lemma \ref{feasible1}. 
This would complete the proof of \eqref{fhlimit2}.

It is left to prove \eqref{zerolim}. 
Because 
\begin{align*} 
\int_{0}^1 F_{\rho}^{-1}(1-z)\dz=\be[\rho]<\infty,
\end{align*}
it follows from the dominated convergence theorem that 
\begin{align} \label{zerolim2}
\lim_{\lam\to\lamup}\int_{z_2(\lam)}^1 F_{\rho}^{-1}(1-z)\dz=0.
\end{align}
Under Assumption \ref{ass1}, there exists some $\lam_0>0$ such that 
\begin{align*} 
\int_{0}^1 (U')^{-1}(\lam_0 F_{\rho}^{-1}(1-z)) F_{\rho}^{-1}(1-z)\dz
= \be[\rho (U')^{-1}(\lam_0\rho)]<\infty.
\end{align*}
Because 
\[ \lim\limits_{\lam\to\lamup }h(\lam,\underz)=+\infty,\quad \lim_{\lam\to\lamup}z_2(\lam)=1,\]
by the monotonicity of $U'$ and the dominated convergence theorem, 
\begin{align*} 
0 &\leq\limsup_{\lam\to\lamup}\int_{z_2(\lam)}^1 (U')^{-1}(h(\lam,\underz) F_{\rho}^{-1}(1-z))) F_{\rho}^{-1}(1-z)\dz\\
&\leq \limsup_{\lam\to\lamup}\int_{z_2(\lam)}^1 (U')^{-1}( \lam_0 F_{\rho}^{-1}(1-z))) F_{\rho}^{-1}(1-z)\dz=0,
\end{align*}
which together with \eqref{zerolim2} leads to \eqref{zerolim}.
\end{proof}

\begin{proof}[The proof of Theorem \ref{main2}]
If $R(\underz)<\barz<C(\underz)$, then by Lemma \ref{lagrangeexist} and continuity, 
there exits $\lam^*\in(0,\lamup)$ such that $$f(\lam^*, h(\lam^*,\underz ))=\barz.$$
Let $\lamb^*=h(\lam^*,\underz )$, then
\[\int_0^1 G^{*}_{\lam^{*},\lamb^{*}}(z)F_{\rho}^{-1}(1-z)\dz=f(\lam^*, h(\lam^*,\underz ))=\barz,\]
and
\[\frac{1}{\alpha}\int_{0}^{\alpha}G^{*}_{\lam^{*},\lamb^{*}}(z)\dz
=g(\lam^*, h(\lam^*,\underz ))=\underz.\]
 Hence, by Lemma \ref{t1}, $G^{*}_{\lam^{*},\lamb^{*}}$ is the optimal solution to the problem \eqref{e4}, completing the proof of Theorem \ref{main2}. 
\end{proof}
Until now, we have finished the study of the quantile optimization problem \eqref{e4}.
We are ready to solve the original stochastic control problem \eqref{e2}.

\subsection{Optimal solution for \eqref{e2}}
Before obtaining the optimal investment strategy for the stochastic control problem \eqref{e2}, we first derive the optimal terminal surplus for \eqref{e3} and the optimal wealth as follows for \eqref{e2}.
\begin{theorem}
[Optimal terminal values for \eqref{e2} and \eqref{e3} with an effective \ltvar constraint] 
\label{main3}
When $R(\underz)<\barz<C(\underz)$, the optimal terminal surplus $Z^*(T)$ for the problem \eqref{e3} is 
\[Z^{*}(T)= 
\begin{cases} 
I(\lamb \rho)&\quad\rho\leq \underline{\rho},\\
I(\lamb \underline{\rho})&\quad\underline{\rho}<\rho\leq\overline{\rho},\\
I(\lamb \rho-\frac{\lam\lamb}{\alpha})&\quad\overline{\rho}<\rho\leq \rho_{\ell},\\
\ell&\quad\rho> \rho_{\ell},
\end{cases}
\]
and the corresponding optimal terminal wealth for the problem \eqref{e2} is 
\[X^{*}(T)= 
\begin{cases} 
I(\lamb \rho)+L(T) &\quad\rho\leq \underline{\rho},\\
I(\lamb \underline{\rho})+L(T)&\quad \underline{\rho}<\rho\leq\overline{\rho},\\
I(\lamb \rho-\frac{\lam\lamb}{\alpha})+L(T)&\quad\overline{\rho}<\rho\leq \rho_{\ell},\\
\ell+L(T)&\quad\rho> \rho_{\ell},
\end{cases}
\]
where $$I(x)=(U')^{-1}(x), \quad \rho_{\ell}=\frac{U'(l)}{\lamb}+\frac{\lam}{\alpha}, \quad \overline{\rho}=\underline{\rho}+\frac{\lam}{\alpha},$$ and 
$\underline{\rho}$
satisfies 
\[\Big(\varphi_{\lam}(x)-\varphi_{\lam}(s(x))-\varphi'_{\lam}(x)(x-s(x))\Big)\Big|_{x=1-F_\rho(\underline{\rho})}=0,\] 
and
\[s(x)=1-F_{\rho}(F_{\rho}^{-1}(1-x)+\frac{\lam}{\alpha}),\]
and $(\lam,\lamb)$ solves
\[\be[\rho Z^{*}(T)]=\barz, \quad \tvar^{-}_{\alpha}(Z^{*}(T))=\esb.\] 
\end{theorem}
\begin{proof} 
Let $G^{*}_{\lam^{*},\lamb^{*}}$ be the optimal solution to the problem \eqref{e4}, given in 
Theorem \ref{main2}. 
Then the optimal surplus $Z^*(T)$ for the problem \eqref{e3} is 
$Z^*(T)=G^{*}_{\lam^{*},\lamb^{*}}(1-F_{\rho}(\rho))$. The claim follows after some simple calculations. 
\end{proof} 

From the expression, we can see that the optimal terminal surplus $Z^{*}(T)$ can be divided into {four} regions corresponding to the four different market scenarios. We can use some terminologies as in \cite{C18}: 
in the good states $\rho<\underline{\rho}$, the optimal surplus $I(\lamb \rho)$ is parallel with the classical Merton's strategy but with a different Lagrange multiplier. In the moderate states $\underline{\rho}<\rho\leq\overline{\rho}$, the optimal surplus remains at a constant $I(\lamb \underline{\rho})$. 
Different from the VaR case, the constant is not equal to the \tvar reference value. 
If the states get worse $\overline{\rho}<\rho\leq \rho_{\ell}$, the optimal surplus $I(\lamb \rho-\frac{\lam\lamb}{\alpha})+L(T)$ is a shifted Merton's strategy with an insurance. Moreover, the reference level impacts the optimal terminal surplus indirectly through the Lagrange multipliers $\lam$ and $\lamb$. The participants choose to partially insure and bear the loss, because it is too costly to insure the loss above the reference level. 
When market conditions deteriorate $\rho> \rho_{\ell}$, a minimum level of the surplus is guaranteed by the PI constraint. 
\begin{figure}[htbp]
\centering
\includegraphics[height=8cm,width=10cm]{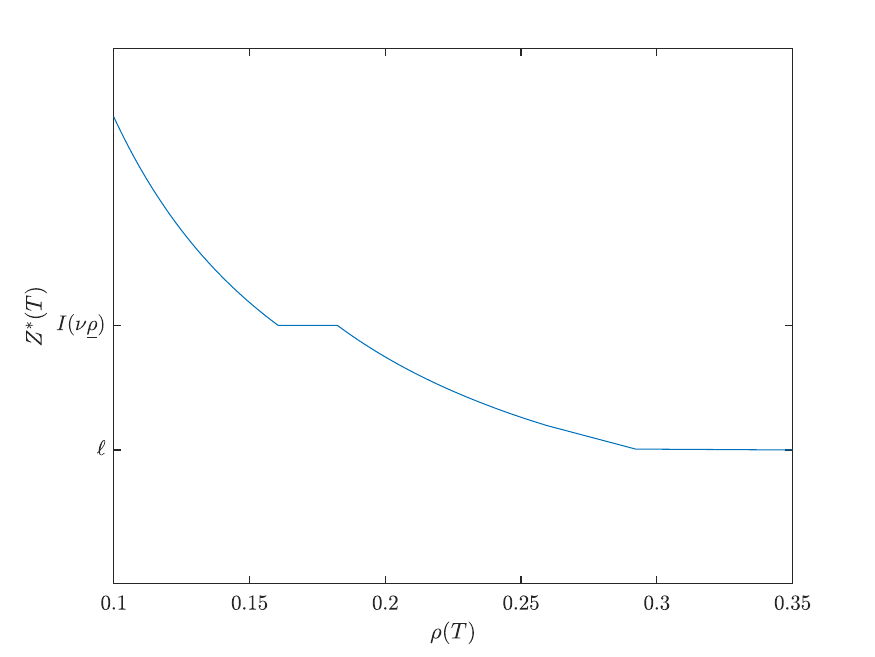}
\caption{The optimal terminal surplus.}
\label{2}
\end{figure}

The relationship between $Z^{*}(T)$ and $\rho$ is illustrated in Figure \ref{2}. As stated in \cite{W21}, \tvar may incur larger losses than the benchmark case, and even larger losses than VaR in case of significant financial market stress, and we can see it in Figure 1 in \cite{W21}. In the extreme loss states, the PI constraint holds, and the optimal surplus process is equal to $\ell$. For the VaR-PI joint risk management problem in \cite{C18}, the optimal wealth profile may has three or four regions under different magnitude of $\ell$. However, under \tvar-PI constraint it has four parts across the states. Moreover, same as the classical case, the optimal terminal surplus is always decreasing with respect to $\rho$. The entire curve of $X^{*}(T)$ is moved up $L(T)$ units from $Z^{*}(T)$ to ensure the optimal terminal wealth lies above the minimal performance $L(T)$. 

\begin{remark} 
The authors \cite{MX21} considers a joint VaR-PI based portfolio selection problem in a behavioral setting. Instructive results are derived even with probability weighting function involved. As a general probability function is introduced to our problem, it is intractable to derive an analytical forms of the counterpart of $\delta_{\lam,\lamb}(\cdot)$, let alone the optimal solutions.
\end{remark}
After obtaining the optimal terminal wealth at the retirement time, we can derive explicit optimal wealth process and investment strategy for our original problem \eqref{e2}.

 \begin{proposition}\label{prop:control}
Assuming the utility function is a CRRA one given by $$u(x)=\frac{x^{1-\gamma}}{1-\gamma},\quad\gamma\in(0,1).$$
Then the optimal wealth process and the optimal investment strategy for the problem \eqref{e2} are given as follows. 
\begin{itemize}
\item 
The optimal wealth process is given by 
\begin{align*}
X^{*}(t)&=(\lamb\rho(t))^{-\frac{1}{\gamma}}e^{\Gamma(t)}\Phi(d_2(\underline{\rho}))+(\lamb\underline{\rho})^{-\frac{1}{\gamma}}e^{-r_0(T-t)}\Big[\Phi(d_1(\overline{\rho}))-\Phi (d_1(\underline{\rho}))\Big]\\
&\quad\;+e^{-r_0(T-t)}H\big(\tfrac{1}{\gamma}\big)+\ell e^{-r_0(T-t)}\Phi(-d_1( \rho_{\ell}))-D(t)+L(t),
\end{align*}
where 
\begin{align*}
d_{1,2}(x) &=\frac{\log\frac{x}{\rho(t)}+(r_0\pm\frac{1}{2}\Vert\xi\Vert^2)(T-t)}{\Vert\xi\Vert\sqrt{T-t}},\\
 \Gamma(t) &=e^{\frac{1-\gamma}{\gamma}(r_0+\frac{1}{2\gamma}\Vert\xi\Vert^2)(T-t)},\\
H(x)&=\int_{d_1(\overline{\rho})}^{d_1( \rho_{\ell})}\phi(z)\Big(\lamb\rho(t)e^{\Vert\xi\Vert\sqrt{T-t}z-(r_0+\frac{1}{2}\Vert\xi\Vert^2)(T-t)}-\frac{\lam\lamb}{\alpha}\Big)^{-x}\dz,
\end{align*}
and 
$\Phi(\cdot)$ denotes the standard normal distribution function and $\phi(\cdot)$ denotes its density function.

\item
The optimal investment strategy is 
\[\pi^*(t)=(\sigma^{-1})^{\top}\Big[\xi\lambda(t)-\sigma_DD(t)+\sigma_LL(t)\Big].\] 
In other words, the amounts of the wealth invested in the inflation-linked bond and the stock are, respectively,\footnote{One can write the strategy as a feedback function of the time $t$ and stock price $S(t)$ and the inflation-linked bond price $B(t)$.
We encourage the diligent readers to do this. 
}
 \begin{align*}
\pi_1^*(t)=&\frac{1}{\sigma_I}\Big[\xi_1\lam(t)-\sigma_Y\rho_{IY}D(t)+\sigma_a\rho_{Ia}L(t)\Big]\\&-\frac{\rho_{IS}}{\sigma_I\sqrt{1-\rho_{IS}^2}}\Big[\xi_2\lambda(t)-\sigma_Y\sqrt{1-\rho_{IY}^2}D(t)+\sigma_a\sqrt{1-\rho_{Ia}^2}L(t)\Big],
\end{align*}
\[\pi_2^*(t)=\frac{1}{\sigma_S\sqrt{1-\rho_{IS}^2}}\Big[\xi_2\lambda(t)-\sigma_Y\sqrt{1-\rho_{IY}^2}D(t)+\sigma_a\sqrt{1-\rho_{Ia}^2}L(t)\Big],\]
where 
\begin{align*}
\lam(t)=&(\lamb\rho(t))^{-\frac{1}{\gamma}}e^{\Gamma(t)}\Big[\frac{\Phi(d_2(\underline{\rho}))}{\gamma }+\frac{\phi(d_2(\underline{\rho}))}{\Vert\xi\Vert \sqrt{T-t}}\Big]-(\lamb\underline{\rho})^{-\frac{1}{\gamma}}e^{-r_0(T-t)}\frac{\phi(d_1(\underline{\rho}))}{\Vert\xi\Vert \sqrt{T-t}}\\&+\frac{1}{\gamma}e^{-r_0(T-t)}\Big[H\big(\tfrac{1}{\gamma}\big)+\frac{\lamb\lam}{\alpha}H\big(\tfrac{\gamma+1}{\gamma}\big)\Big].
\end{align*} 

\end{itemize}
\end{proposition}
\begin{proof}
Because $\rho(t)Z^{*}(t)$ is a martingale,
\begin{align*}
Z^{*}(t)=\frac{1}{\rho(t)}\be_t[\rho(T) Z^{*}(T)] &=\frac{1}{\rho(t)}\be_t\Big[\rho(T)[(\lamb\rho(T))^{-\frac{1}{\gamma}}I_{\{\rho(T)\leq \underline{\rho}\}}+(\lamb\underline{\rho})^{-\frac{1}{\gamma}}I_{\{\underline{\rho}<\rho(T)\leq\overline{\rho}\}}\\
&\qquad\qquad\qquad+(\lamb\rho(T)-\frac{\lam\lamb}{\alpha})^{-\frac{1}{\gamma}}I_{\{\overline{\rho}<\rho(T)\leq \rho_{\ell}\}}+\ell I_{\{\rho>\rho_{\ell}\}}]\Big]
\end{align*}
Because the pricing kernel follows geometric Brownian motion, 
 \[\frac{\log\frac{\rho(T)}{\rho(t)}+(r_0+\frac{1}{2}\Vert\xi\Vert^2)(T-t)}{\Vert\xi\Vert\sqrt{T-t}}\sim N(0,1).\] 
So the first part can be evaluated as 
\begin{align*}
&\frac{1}{\rho(t)}\lamb^{-\frac{1}{\gamma}}\be_t[(\rho(T))^{\frac{\gamma-1}{\gamma}}I_{\{\rho(T)\leq \underline{\rho}\}}]\\
&=(\lamb\rho(t))^{-\frac{1}{\gamma}}e^{\frac{1-\gamma}{\gamma}(r_0+\frac{1}{2}\Vert\xi\Vert^2)(T-t)}\int_{-\infty}^{d_1(\underline{\rho})+\Vert\xi\Vert\sqrt{T-t}}
e^{-\frac{1-\gamma}{\gamma}x\Vert\xi\Vert\sqrt{T-t}}\frac{1}{\sqrt{2\pi}}e^{-\frac{x^2}{2}}\dx\\
&=(\lamb\rho(t))^{-\frac{1}{\gamma}}e^{\frac{1-\gamma}{\gamma}(r_0+\frac{1}{2\gamma}\Vert\xi\Vert^2)(T-t)}\int_{-\infty}^{d_2(\underline{\rho})}\frac{1}{\sqrt{2\pi}}e^{-\frac{y^2}{2}}dy\\
&=(\lamb\rho(t))^{-\frac{1}{\gamma}}e^{\Gamma(t)}\Phi(d_2(\underline{\rho})).
\end{align*}
The remaining expectations can be calculated similarly. They lead to the desired expression for $X^*(t)$ by recalling that 
$X^*(t)=Z^*(t)+L(t)-D(t).$

The above result also shows that $Z^*(t)=F(t,\rho(t))$ for an explicit deterministic function $F$. Applying Ito's lemma, we have
\[\dd Z^*(t)=(\cdots) \dt+\frac{\partial{F}}{\partial{\rho}}(-\rho(t)\xi^{\top}\dd W(t)). \]
Comparing to the diffusion term in \eqref{dynamicz} yields 
\[-\frac{\partial{F}}{\partial{\rho}}\rho(t)\xi^{\top}=\big(\sigma^{\top}\pi^*(t)+\sigma_DD(t)-\sigma_LL(t)\big)^{\top},\]
so 
\[\pi^*(t)=-(\sigma^{-1})^{\top}\Big[\xi\frac{\partial{F}}{\partial{\rho}}\rho(t)+\sigma_DD(t)-\sigma_LL(t)\Big].\] 
A tedious calculation shows that 
\begin{align*}
\frac{\partial F}{\partial \rho}\rho(t)&=-(\lamb\rho(t))^{-\frac{1}{\gamma}}e^{\Gamma(t)}\Big[\frac{\Phi(d_2(\underline{\rho}))}{\gamma }+\frac{\phi(d_2(\underline{\rho}))}{\Vert\xi\Vert \sqrt{T-t}}\Big]+(\lamb\underline{\rho})^{-\frac{1}{\gamma}}e^{-r_0(T-t)}\frac{\phi(d_1(\underline{\rho}))}{\Vert\xi\Vert \sqrt{T-t}}\\
&\quad\;-\frac{1}{\gamma}e^{-r_0(T-t)}\Big[H\big(\overline{\rho},\tfrac{1}{\gamma}\big)+\tfrac{\lamb\lam}{\alpha}H\big(\overline{\rho},\tfrac{\gamma+1}{\gamma}\big)\Big].
\end{align*}
Combining above completes the proof. 
\end{proof}
Clearly the optimal investment strategy consists of three parts, the latter two of which reflect the impacts of salary and the minimum performance constraint, respectively.

\section{Numerical analysis}\label{sec:numerical}
In this section we present a sensitivity analysis and provide economic interpretations for the optimal surplus value and strategy.

The parameters we use are as follows. 
\begin{gather*}
r=0.02,\ r_0=0.05,\ T=40, \ T'=60,\ \mu_I=0.033,\ \mu_S=0.4,\\
 \mu_Y=0.1,\ \mu_a=0.1,\ \sigma_I=0.2,\ \sigma_S=0.4,\ \rho_{IS}=0.5,\ \rho_{IY}=0.6,\ \rho_{Ia}=0.55. 
\end{gather*}
 The volatility rate of salary and minimal performance security are $\sigma_Y=0.25$ and $\sigma_a=0.36$, respectively, corresponding to the high risk case in \cite{C17}. The market prices of risk are $\xi_1=0.015$ and $\xi_2=0.035$. The risk aversion parameter is $\gamma=0.8$. Assume the initial salary level is $y_0=1$, the contribution rate is $c=8\%$, and the operation cost is $\ell=30$. After a straightforward calculation, we can get $d_0=8$ and $\ell_0=7$. We choose $\barz=10$ to guarantee that the \ltvar constraint binds. 
\begin{figure}[htbp]
\centering
\includegraphics[height=8cm,width=10cm]{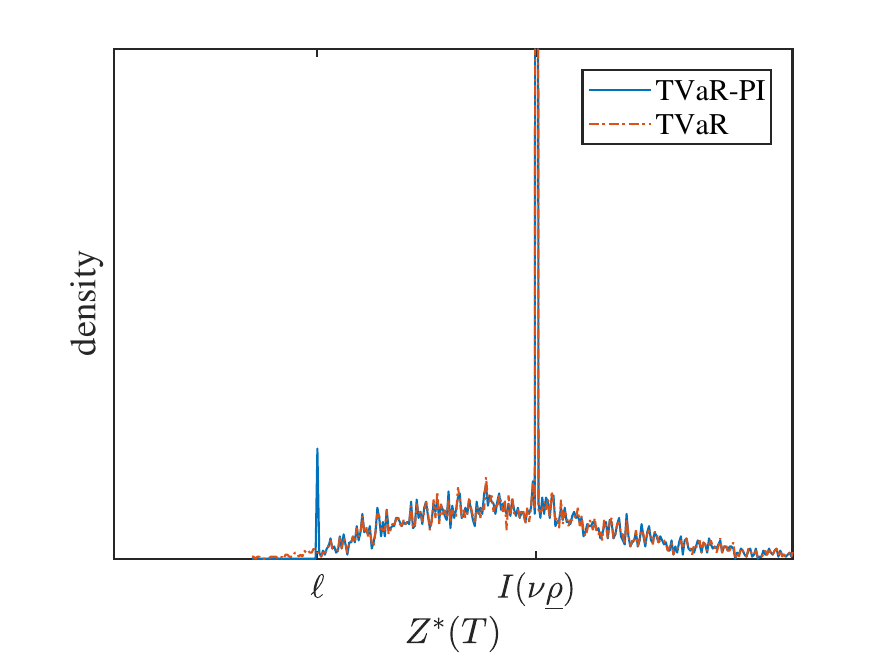}
\caption{Probability density of the optimal surplus.}
\label{3}
\end{figure}

We compare the probability densities of $Z^{*}(T)$ for our \tvar-PI problem and for the pure \tvar problem Figure $\ref{3}$. The distributions of the optimal terminal value of surplus for both problems are continuous. We can see the additional PI constraint can improve the distribution of $Z^{*}(T)$ by setting a lower bound $\ell$, which implies it significantly reduces the tail risk in the worse market scenarios.

\begin{figure}[htbp]
\centering
\includegraphics[height=8cm,width=10cm]{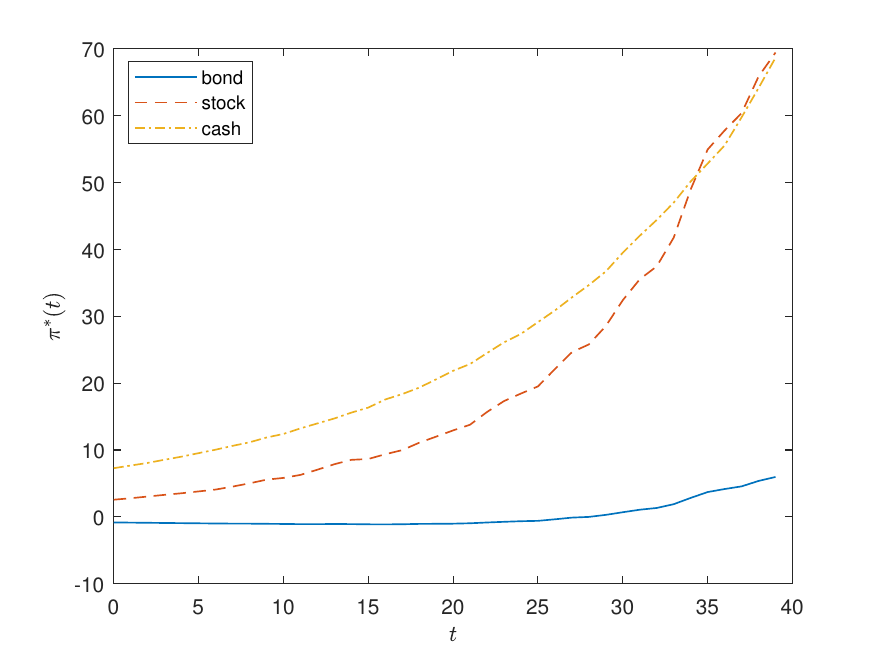}
\caption{Optimal strategy when $\alpha=0.1, \underz=50$.}
\label{4}
\end{figure}

Figure $\ref{4}$. displays the optimal investment strategy for different assets when $\alpha=0.1,\ \underz=50$. The result is obtained by Monte Carlo Methods (MCM). In practice, the confidence level is usually set at a tiny number, for example, $\alpha=0.1\%$ in the insurance regulation framework of Solvency ${\rm \uppercase\expandafter{\romannumeral2}}$. Here we choose $\alpha=0.1$ in order to explore the impact of \ltvar constraint intuitively later.
As is shown in Figure $\ref{4}$, the participant invests a large amount of wealth in the cash at initial time, and this value then increases as time goes on. The stock investment increases and exceeds the cash one after time 33. The member takes a short position in the bond at time 0 and changes to a long position later. That is because in the high risk case, the member faces the \tvar and the minimal performance constraint. Thus a fine-tune is made on the wealth to comply with the both constraints at retirement. In any case, the bond investment is relatively low compared to the other two assets. 

Numerous studies have been conducted to explore the impacts of the economic parameters on the optimal strategy, see \cite{C17}. Thus we examine the role of \ltvar constraint on the problem at hand next.
Still applying the Monte Carlo Methods, we describe the optimal proportions of the total wealth invested in the bond and the stock, i.e. $\frac{\pi^*(t)}{x^*(t)}$, at time $t=\frac{T}{2}$.

Figure $\ref{conf}$ depicts the effect of the confidence level parameter $\alpha$ on the investment strategies. As $\alpha$ rises from 0.1 to 0.2, the proportion of the stock investment is getting smaller slightly. \tvar measures the risk by averaging all VaRs above a confidence level. A higher $\alpha$ means less effort that needs to pay by the member to achieve his goal at retirement. Hence, he invests a smaller proportion in the stock. In the meanwhile, there is an uptrend in the proportion of the inflation-linked bond. That is because investing only in the risk-free asset is hard to satisfy the \tvar and PI constraints at retirement, he has to invest in the inflation-linked bond to protect against the inflation risk.

\begin{figure}[htbp]
\centering
\subfloat[]{
 \begin{minipage}[t]{0.48\textwidth}
 \centering
 \includegraphics[height=6cm,width=8cm]{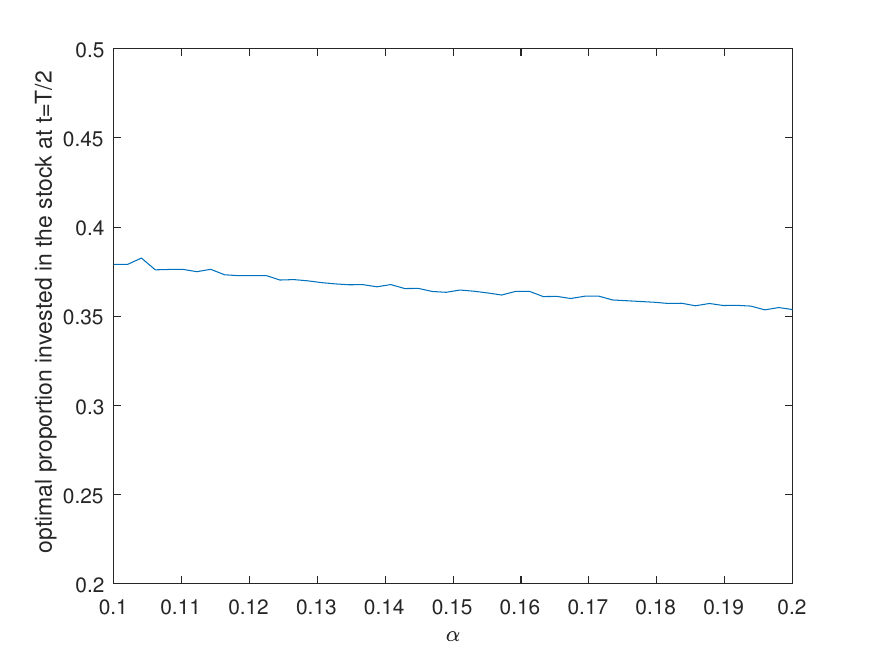} 
\end{minipage}}
\subfloat[]{
 \begin{minipage}[t]{0.48\textwidth}
 \centering
 \includegraphics[height=6cm,width=8cm]{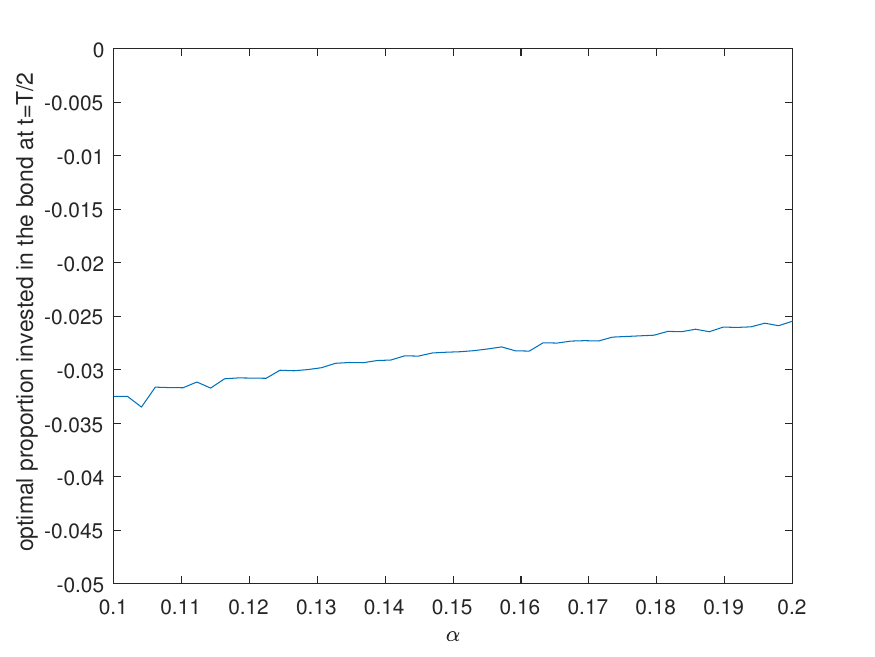}
\end{minipage}}
\caption{Effect of $\alpha$ on the optimal proportions in the stock (a) and bond (b).}
\label{conf}
\end{figure}
We show the impact of \tvar reference level $\underz$ on the investment strategies in Figure $\ref{reference}$. When $\underz $ increases, it is more difficult to satisfy the \ltvar constraint. In this case, he has to pay more attention to the stock to achieve a higher return. As is shown in Figure $\ref{reference}$, the stock as a percentage of total wealth has becoming larger as $\underz$ increases, and the bond is just in the opposite trend. 
 \begin{figure}[htbp]
\centering
\subfloat[]{
 \begin{minipage}[t]{0.48\textwidth}
 \centering
 \includegraphics[height=6cm,width=8cm]{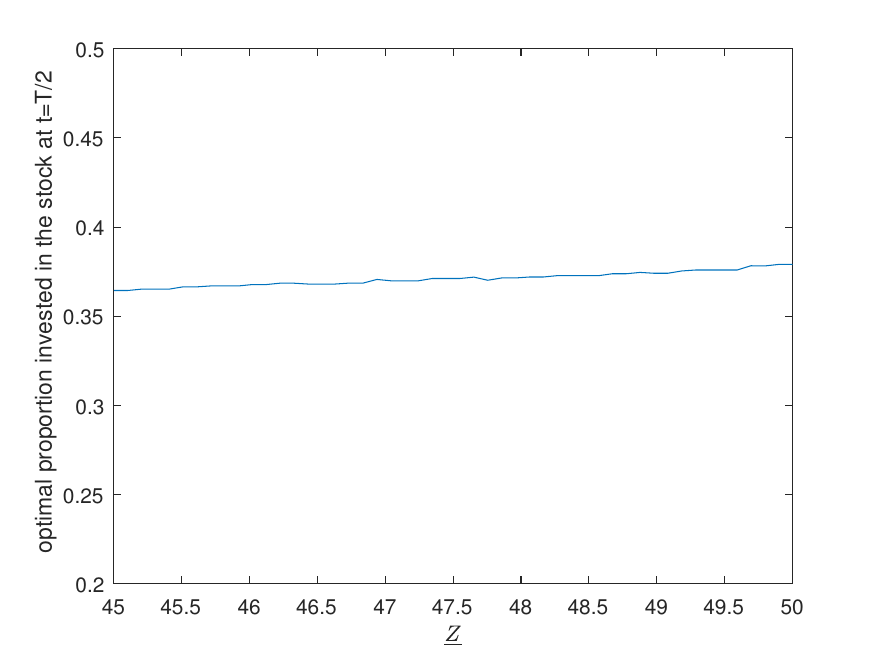} 
\end{minipage}}
\subfloat[]{
 \begin{minipage}[t]{0.48\textwidth}
 \centering
 \includegraphics[height=6cm,width=8cm]{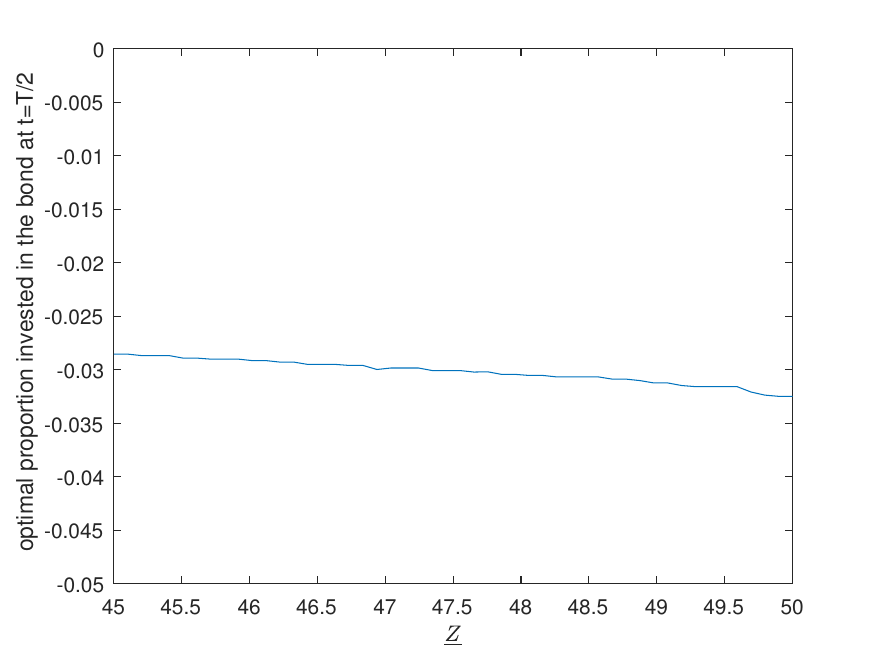}
\end{minipage}}
\caption{Effect of $\underline{z}$ on the optimal proportions in the stock (a) and bond (b).}
\label{reference}
\end{figure}
\section{Conclusion}\label{sec:conclude}
In this paper, we considered a risk management problem for DC pension plan under joint \tvar and PI constraints. The pension member can invest in the cash, bond and stock to protect against the inflation risk and the salary risk. Based on the techniques of quantile formulation and martingale method, we derived closed-form optimal terminal wealth and optimal investment strategies. The optimal terminal wealth turns out to be a piecewise smooth function of the market price of risk. Furthermore, we analyzed the impacts of \ltvar constraint on the optimal investment strategies via numerical study. The numerical results indicate that the PI constraint can improve the risk management. \par
The present work may be extended further. For example, the representing member in this paper is assumed to be risk averse, in reality, however, many people are risk-seeking in the loss situation. Therefore, it makes sense to introduce the so-called $S$-shaped utility function to the model, which will be left for study in the future.\footnote{There exists portfolio selection literature that investigates the effectiveness of commonly imposed risk constraints including \tvar in loss-averse traders, see \cite{AB19}. 
However, a meticulous analysis on the behaviors of an $S$-shaped member is still lacking.}

\normalem 

\end{document}